\documentclass[aps,pra,twocolumn,nofootinbib, superscriptaddress,10pt]{revtex4-2}

\usepackage{graphicx}
\usepackage{epstopdf}
\usepackage{amsmath}
\usepackage{amssymb}
\usepackage{mathrsfs}
\usepackage{amsthm}
\usepackage{bm}
\usepackage{url}
\usepackage[T1]{fontenc}
\usepackage{csquotes}
\MakeOuterQuote{"}
\usepackage{thmtools, thm-restate}

\newtheoremstyle{note}
  {\topsep/2}               
  {\topsep/2}               
  {}                      
  {\parindent}            
  {\itshape}              
  {.}                     
  {5pt plus 1pt minus 1pt}
  {}

\theoremstyle{note}

\newtheorem{lemma}{Lemma}
\newtheorem{conjecture}{Conjecture}
\newtheorem{corollary}{Corollary}
\newtheorem{proposition}{Proposition}

\theoremstyle{definition}

\theoremstyle{remark}

\def\vec#1{\bm{#1}} 

\newcommand{\tr}{\operatorname{tr}}

\newcommand{\spa}{\operatorname{span}}

\newcommand{\imply}{\mathrel{\Rightarrow}}

\newcommand{\lsp}{\hspace{0.1em}}

 \newcommand{\rme}{\mathrm{e}}
 \newcommand{\rmi}{\mathrm{i}}

 \newcommand{\rmC}{\mathrm{C}}

 \newcommand{\rmR}{\mathrm{R}}
 \newcommand{\rmT}{\mathrm{T}}

  \newcommand{\hs}{\mathrm{HS}}

 \newcommand{\caB}{\mathcal{B}}
 
 \newcommand{\caH}{\mathcal{H}}

 \newcommand{\scrD}{\mathscr{D}}
 \newcommand{\scrI}{\mathscr{I}}
 \newcommand{\scrK}{\mathscr{K}}

 \newcommand{\scrP}{\mathscr{P}}
 \newcommand{\scrS}{\mathscr{S}}
 \newcommand{\scrU}{\mathscr{U}}

\newcommand{\be}{\begin{equation}}
\newcommand{\ee}{\end{equation}}
\newcommand{\ba}{\begin{align}}
\newcommand{\ea}{\end{align}}

\def\<{\langle}  
\def\>{\rangle}  

\def\eqref#1{\textup{(\ref{#1})}}  
\newcommand{\eref}[1]{Eq.~\textup{(\ref{#1})}}
\newcommand{\Eref}[1]{Equation~\textup{(\ref{#1})}}

\newcommand{\esref}[2]{Eqs.~\textup{(\ref{#1})} and \textup{(\ref{#2})}}

\newcommand{\fref}[1]{Fig.~\ref{#1}}

\newcommand{\sref}[1]{Sec.~\ref{#1}}
\newcommand{\Sref}[1]{Section~\ref{#1}}

\newcommand{\thref}[1]{Theorem~\ref{#1}}
\newcommand{\Thref}[1]{Theorem~\ref{#1}}
\newcommand{\thsref}[1]{Theorems~\ref{#1}}
\newcommand{\Thsref}[1]{Theorems~\ref{#1}}

\newcommand{\lref}[1]{Lemma~\ref{#1}}
\newcommand{\Lref}[1]{Lemma~\ref{#1}}
\newcommand{\lsref}[1]{Lemmas~\ref{#1}}

\newcommand{\pref}[1]{Proposition~\ref{#1}}
\newcommand{\Pref}[1]{Proposition~\ref{#1}}
\newcommand{\psref}[1]{Propositions~\ref{#1}}
\newcommand{\Psref}[1]{Propositions~\ref{#1}}

\newcommand{\crref}[1]{Corollary~\ref{#1}}
\newcommand{\Crref}[1]{Corollary~\ref{#1}}
\newcommand{\crsref}[1]{Corollaries~\ref{#1}}

\newcommand{\cref}[1]{Conjecture~\ref{#1}}
\newcommand{\Cref}[1]{Conjecture~\ref{#1}}

\newcommand{\aref}[1]{Appendix~\ref{#1}}

\newcommand{\rcite}[1]{Ref.~\cite{#1}}
\newcommand{\rscite}[1]{Refs.~\cite{#1}}

\begin{document}
	\title{Hiding and masking quantum information in complex and  real quantum mechanics}

	\author{Huangjun Zhu}
	\email{zhuhuangjun@fudan.edu.cn}
	
	\affiliation{State Key Laboratory of Surface Physics and Department of Physics, Fudan University, Shanghai 200433, China}

	\affiliation{Institute for Nanoelectronic Devices and Quantum Computing, Fudan University, Shanghai 200433, China}
	
	\affiliation{Center for Field Theory and Particle Physics, Fudan University, Shanghai 200433, China}

	\begin{abstract}
Classical information can be completely hidden in the correlations of bipartite quantum systems. However, it is impossible 
to hide or mask all quantum information according to  the no-hiding and no-masking theorems derived recently. Here we show that any set of informationally complete quantum states is neither hidable nor maskable, thereby strengthening the  no-hiding and no-masking theorems known before. Then,  by virtue of Hurwitz-Radon matrices (representations of the Clifford algebra), we show that information about real quantum states can be completely hidden in the correlations, although the minimum dimension of the composite Hilbert space  required increases exponentially with the dimension of the original Hilbert space. Moreover, the set of real quantum states is a maximal maskable set within quantum theory and has a surprising connection with maximally entangled states. These results offer valuable insight on the potential and limit of hiding and masking quantum information, which are of intrinsic interest to a number of active research areas. 
	\end{abstract}

	\date{\today}
	\maketitle
	

\section{Introduction}	
Hiding information in correlations is a simple way of realizing secret sharing \cite{Blak79,Sham79}, which is a primitive to cryptography and secure multiparty computation. In the quantum world, a similar idea is of interest to a wide spectrum of research topics, including but not limited to quantum secret sharing \cite{HillBB99,ClevGL99,DiViLT02}, quantum communication, information scrambling, and the black-hole information paradox \cite{Page93,HaydP07,SekiS08,LiuLZZ18L}.	
However,  it is impossible to hide or mask all quantum information in correlations according to the no-hiding theorem   \cite{BrauP07}  and no-masking theorem \cite{ModiPSS18}  derived recently, in sharp contrast with classical information. These no-go theorems are reminiscent of the no-cloning theorem  \cite{WootZ82,Diek82,LamaSHB02} and  no-broadcasting theorem~\cite{BarnCFJ96}, which play crucial roles in quantum cryptography.
On the other hand, little is known about hiding or masking quantum information in restricted sets of quantum states \cite{LianLF19,DingH20,DuGCH21}. An example of special interest is the set of 
real quantum states as represented by real density matrices \cite{Arak80,McKaMG09,HardW12,Baez12,RenoTWT21}, which is the starting point of the resource theory of imaginarity \cite{HickG18,WuKRS21}.

Here we show that it is impossible to hide or mask any set of quantum states that is \emph{informationally complete} (IC), thereby strengthening the  no-hiding and no-masking theorems and establishing an information theoretical underpinning of these no-go results.  This conclusion can serve as the common starting point for deriving and strengthening a number of results on quantum information masking. As implications, a set of qubit states is hidable or maskable iff  the corresponding set of Bloch vectors is contained in a disk \cite{ModiPSS18, LianLF19,DingH20}.
In addition, 
it is impossible to hide or mask any set of quantum states that has a nonzero measure; previously, quite restrictive assumptions are required to derive 
a similar result  \cite{LianLFF20}. Furthermore, it is impossible to hide or mask any set of quantum states that can form a (weighted complex projective) 2-design.

 By virtue of \emph{Hurwitz-Radon} (HR) matrices \cite{Hurw23,Rado22,Eckm06}, we further show that information about real quantum states can be completely hidden in the correlations. 
In addition, we determine the minimum dimension and entanglement cost required to achieve this task. It turns out the minimum dimension increases exponentially with the dimension of the original Hilbert space. Moreover, the set of real quantum sates is a maximal maskable set within quantum theory. Meanwhile, there is a simple connection between the concurrence \cite{HillW97,Woot98,RungBCH01} of the output state of any masker for real quantum states and the robustness of imaginarity of the input state  \cite{HickG18,WuKRS21}.
Our study offers valuable insight on the potential and limit of hiding and masking quantum information. It may shed light on a number of active research areas, including quantum secret sharing \cite{HillBB99,ClevGL99,DiViLT02}, information scrambling,  black-hole information paradox \cite{Page93,HaydP07,SekiS08,LiuLZZ18L},  resource theory of imaginarity \cite{HickG18,WuKRS21}, and foundational studies on  quantum mechanics \cite{Woot86,Hard01,ChirDP11}. 

The rest of this paper is organized as follows. In \sref{sec:HideMask} we review the basic ideas of hiding and masking quantum information and introduce the concepts of masking spectrum, masking purity, entanglement of masking, and maximal maskable sets. In \sref{sec:LimHideMask} we prove that it is impossible to hide or mask any set of quantum states that is IC. In \sref{sec:HR} we discuss the properties and construction of HR matrices. In \sref{sec:MaskRQ} we prove that the set of real quantum states is maskable and is a maximal maskable set. In \sref{sec:HideNoMask} we construct a hidable set that is not maskable. In \sref{sec:PhaseReal} we clarify the relation between real states and phase-parameterized states. \Sref{sec:sum} summarizes this paper. To streamline the presentation, most technical proofs are relegated to the Appendix.

\section{\label{sec:HideMask}Hiding and masking  quantum information}
Let $\caH$ be a $d$-dimensional Hilbert space with $d\geq 2$ and  let $\{|j\>\}_{j=0}^{d-1}$ be the computational basis. Denote by $\scrD(\caH)$ the set of all quantum states (represented by density matrices) on $\caH$ 
and by  $\scrP(\caH)$ the set of all pure states (represented by  rank-1 projectors). Denote by $\scrD^\rmR(\caH)$ the set of  real quantum states (represented by real density matrices with respect to the computational basis) and  by  $\scrP^\rmR(\caH)$ the set of real pure states.

Let $M$ be
an isometry from $\caH$ to a bipartite Hilbert space $\caH_A\otimes \caH_B$ of dimension $d_A\times d_B$. Let  $\rho \in  \scrD(\caH)$, 
\begin{equation}
\varrho=M(\rho):=M\rho M^\dag,\;\; \varrho_A=\tr_B(\varrho), \;\;  \varrho_B=\tr_A(\varrho).  \label{eq:varrhoAB}
\end{equation} 
Given $\scrS\subseteq \scrD(\caH)$, define 
\begin{equation}
M(\scrS):=M\scrS M^\dag=\{\varrho| \rho\in \scrS\}.
\end{equation}
If the reduced state $\varrho_A$ is independent of $\rho$ in $\scrS$, then $M$ can hide the  information of quantum states in the set $\scrS$ from subsystem $A$ (cf. \fref{fig:HideMask}) and is called a \emph{partial masker} (here we only consider partial maskers for subsystem $A$). So it is natural to expect that the information about the original states spreads over the correlations. However, this possibility is ruled by the no-hiding theorem  when $\scrS=\scrD(\caH)$ or $\scrS=\scrP(\caH)$. In this case,  any partial masker can only hide the information in a trivial way by transferring the information to another subsystem \cite{BrauP07}. This surprising fact has implications for many active research topics, including the black-hole information paradox. However, little is known about this issue when $\scrS$ is a smaller subset.

The set $\scrS$ is \emph{hidable} if it consists of only one quantum state or if there exists a nontrivial partial masker. Otherwise, the set $\scrS$ is  \emph{antiscrambling}. In the latter case, $\scrS$ contains at least two distinct quantum states and every partial masker $M:\caH\mapsto \caH_A\otimes \caH_B$ for $\scrS$ is trivial in this  sense: there exists a subspace $\caH_B'=\caH_{B_1}\otimes\caH_{B_2}$ in $\caH_B$ such that 
\begin{equation}\label{eq:antiscrambling}
M(\rho) =\varrho_{AB_1}\otimes \varrho_{B_2} \quad \forall \rho\in \scrS,
\end{equation}
where $\varrho_{AB_1}$ is a fixed density matrix on $\caH_A\otimes \caH_{B_1}$, which is independent of $\rho$ in $\scrS$, while $\varrho_{B_2}$ is a density matrix on $\caH_{B_2}$ that depends on $\rho$. In the current language, the no-hiding theorem \cite{BrauP07} states that $\scrP(\caH)$ and $\scrD(\caH)$ are antiscrambling.

\begin{figure}
	\includegraphics[width=6.5cm]{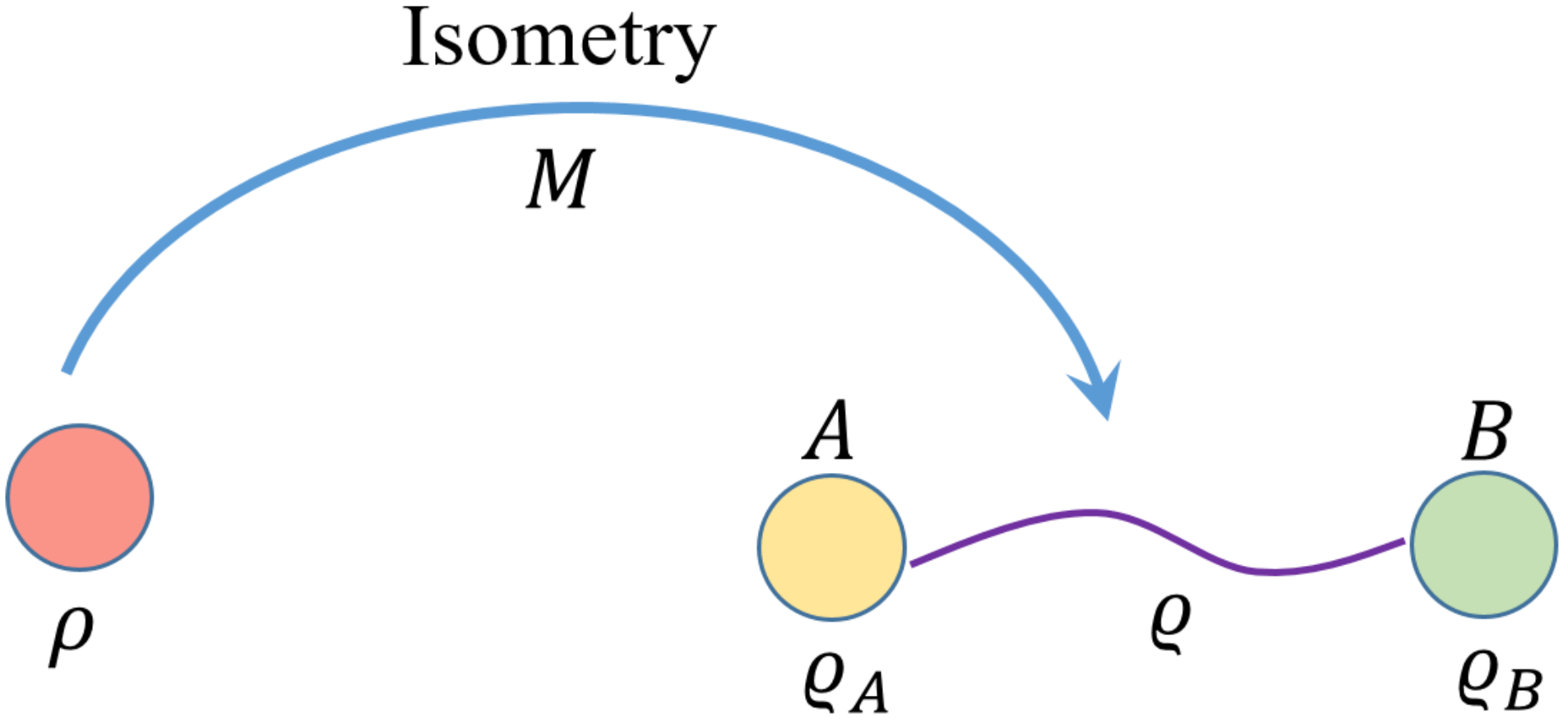}
	\caption{\label{fig:HideMask}
		Hiding and masking of quantum information. In the former case, $\varrho_A$ is independent of $\rho$ (in a given set $\scrS$); in the latter case, both $\varrho_A$ and $\varrho_B$ are independent of $\rho$. }
\end{figure}

By contrast, $M$ can \emph{mask} the set $\scrS$ if both  $\varrho_A$ and $\varrho_B$ are independent of $\rho$ (cf.~\fref{fig:HideMask}); that is,
\begin{equation}
\tr_B[M(\rho)]=\tau_A, \quad \tr_A[M(\rho)]=\tau_B\quad \forall \rho\in \scrS, 
\end{equation}
where $\tau_A$ and $\tau_B$ are fixed density matrices and are referred to as the common reduced density matrices. In this case, $M$ is called a \emph{masker} for $\scrS$, and $\scrS$ is \emph{maskable}. Note that  the  masker $M$ can hide the information from both $A$ and $B$, so all information is hidden 
in the correlations.  Meanwhile, this masker offers a (2,2) threshold scheme for quantum secret sharing \cite{ClevGL99}. A maskable set $\scrS$  is \emph{maximal} if it is not contained in any other maskable set in $\scrD(\caH)$. When  $\scrS\subseteq \scrD^\rmR(\caH)$, the masker $M$ is called a real masker if $M(\scrS)\subseteq \scrD^\rmR(\caH_A\otimes \caH_B)$.

The \emph{masking spectrum} of $M$ is defined as the nonzero spectrum of $\tau_A$ or $\tau_B$, assuming that $\scrS$ contains at least one pure state, so that $\tau_A$ and $\tau_B$  have the same nonzero spectrum. 
The \emph{masking purity} $\wp$ of $M$ refers to the purity of $\tau_A$ or $\tau_B$, that is, 
\begin{align}\label{eq:MaskPurity}
\wp:=\tr(\tau_A^2)=\tr(\tau_B^2).
\end{align}
Given a bipartite entanglement monotone $E$ \cite{HoroHHH09}, the \emph{entanglement of masking} of $\scrS$  is defined as 
\begin{equation}
E(\scrS):=\min_{M}E(\scrS,M),\quad E(\scrS,M):=\max_{\rho\in \scrS} E(M(\rho)),
\end{equation}
where the minimization is over all maskers $M$ for $\scrS$.   The  entanglement of masking associated with real maskers can be defined in a similar way and is denoted by $E^\rmR(\scrS)$. 
 The following lemma is proved in \aref{sec:MaskEntLemProof}. 
\begin{restatable}{lemma}{lemMaskEnt}\label{lem:MaskEnt}
	Suppose $\scrS$ contains a pure state $\rho_0$ and $M$ is a masker for $\scrS$. Then $E(\scrS,M)=E(M(\rho_0))$ for any entanglement monotone $E$.
\end{restatable}


If there exists no masker for $\scrS$, then  $\scrS$ is  not maskable.  
 By  \rcite{ModiPSS18}, the sets $\scrP(\caH)$ and $\scrD(\caH)$ are not maskable.  This no-masking theorem also follows from the no-hiding theorem \cite{BrauP07}; note that a maskable set is hidable (not antiscrambling) by definition, and an antiscrambling set is not maskable. Nevertheless, a maskable set can contain infinite nonorthogonal quantum states. In particular, any set of phase-parameterized states is maskable \cite{ModiPSS18} (cf. \sref{sec:PhaseReal}). In addition, any set of commuting density matrices is maskable \cite{DuGCH21} since these density matrices can be expressed as convex mixtures of pure states in an orthonormal basis and so can be masked using a generalized Bell basis. Together with \rcite{BarnCFJ96}, this fact implies  that any set of quantum  states that can be broadcast is maskable, as illustrated in \fref{fig:NoGo}.

\begin{figure}
	\includegraphics[width=6.5cm]{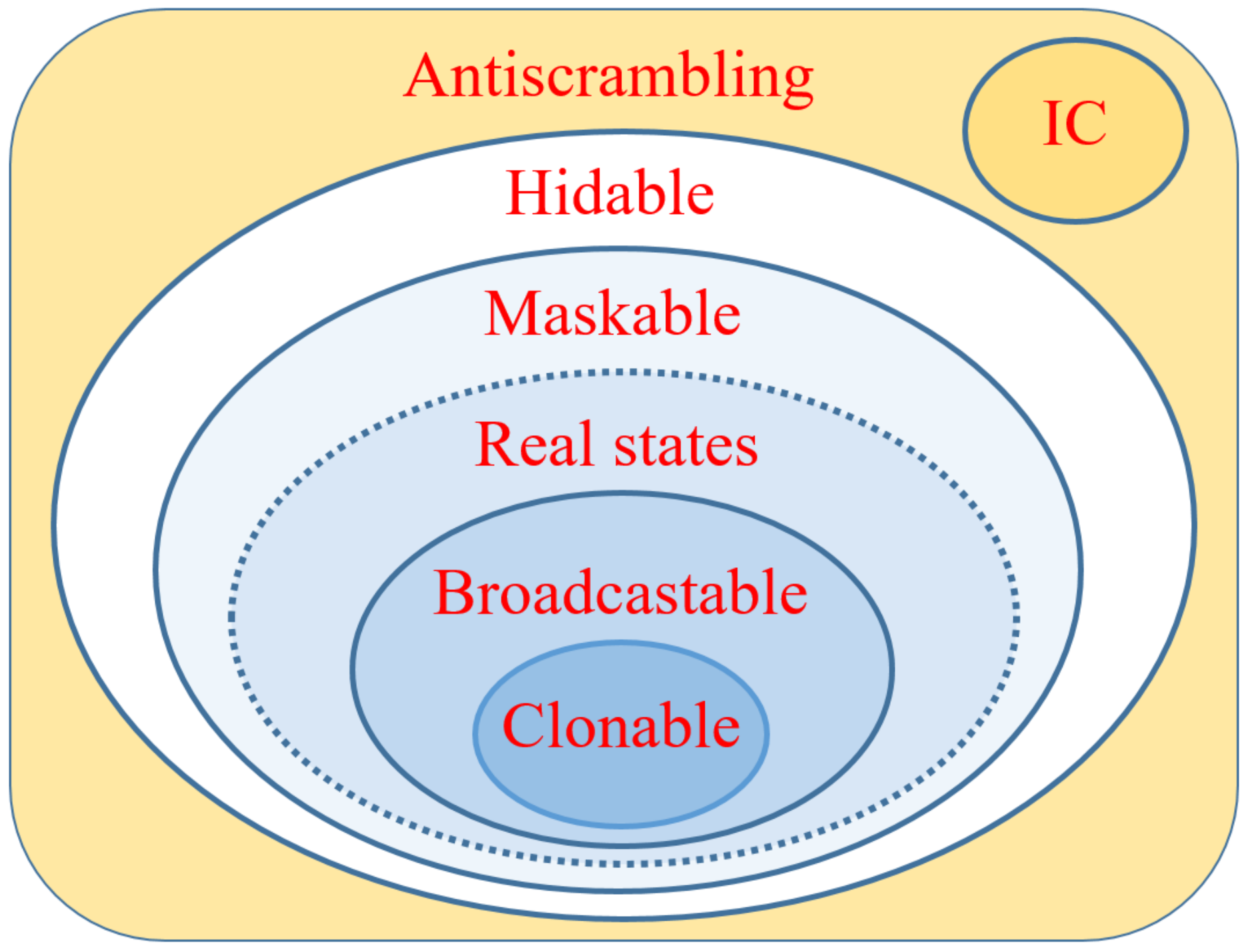}
	\caption{\label{fig:NoGo}
		Hierarchy of no-go theorems and associated sets  of quantum states.	A set of quantum states can be cloned iff distinct states are orthogonal, while it can be broadcast iff all states commute with each other \cite{WootZ82,Diek82,BarnCFJ96}. Any set that can be broadcast is equivalent to a subset of the set $\scrD^\rmR(\caH)$ of real states, which can be masked even with a real masker. Any maskable set can hide from one party in a nontrivial way. So no-hiding and no-masking theorems do not apply to real quantum mechanics. By contrast, any IC set is antiscrambling (cannot hide from one party in a nontrivial way). Note that the set marked by "Antiscrabmling" is the complement of the set marked by "Hidable", as indicated by the color coding.}
\end{figure}

\section{\label{sec:LimHideMask}Limitations on hiding and masking  quantum information}

\subsection{Stronger no-hiding and no-masking theorems}

Let $\scrS\subseteq \scrD(\caH)$ be a set of density matrices. The set $\scrS$ is IC if $\spa(\scrS)=\spa(\scrD(\caH))$. In this case, any other density matrix is uniquely determined by the overlaps (transition probabilities) with density matrices in the set. 
This definition is motivated by an analogous definition for quantum measurements, usually represented by  positive operator-valued measures (POVMs) \cite{Prug77,Scot06,ZhuE11}. By virtue of this concept, we can strengthen the no-hiding and no-masking theorems \cite{BrauP07,ModiPSS18}, which is instructive to understanding the distinctions between complex quantum mechanics and real quantum mechanics. \Pref{pro:MaskSpan} below follows from the very definitions.

\begin{proposition}\label{pro:MaskSpan}
	If $\spa(\scrS_1)\subseteq\spa(\scrS_2)\subseteq \scrD(\caH)$, then $\scrS_1$ is maskable whenever $\scrS_2$ is. If $\spa(\scrS_1)=\spa(\scrS_2)$; then $\scrS_1$ is hidable or maskable  iff $\scrS_2$ is.
\end{proposition}

\begin{restatable}{theorem}{thmHideIC}\label{thm:HideIC}
	Any IC set of quantum states is antiscrambling. 
\end{restatable}
\Thref{thm:HideIC} follows from \rcite{BrauP07} and \pref{pro:MaskSpan}.  It implies \thref{thm:MaskIC} below, which  also
follows from \rcite{ModiPSS18} and \pref{pro:MaskSpan}. Self-contained proofs can be found in \aref{sec:NoHMproof}.

\begin{restatable}{theorem}{thmMaskIC}\label{thm:MaskIC}
	Any IC set of quantum states is not maskable.
\end{restatable}

The implications of \thsref{thm:HideIC} and \ref{thm:MaskIC} are discussed in more detail in the next subsection. Here it should be noted  that \thsref{thm:HideIC} and \ref{thm:MaskIC} do not apply to real quantum mechanics, which is tied to the failure of local tomography \cite{Arak80}, as discussed in \aref{sec:NoHMproof}.

\subsection{Implications}

To understand the implications of \thsref{thm:HideIC} and \ref{thm:MaskIC} concerning the limit of hiding and masking quantum information, it is instructive to  point out a basic property of IC sets of quantum states. Note that  
each density matrix on $\caH$   can be expressed as  $\rho=(\openone/d)+\sum_{j=1}^{d^2-1} r_j E_j$, where $\openone$ is the identity, $\{E_j\}_{j=1}^{d^2-1}$ is an orthonormal basis in the space of traceless hermitian matrices, and $\bm{r}=(r_j)_{j=1}^{d^2-1}$  is a generalized Bloch vector. The affine dimension of a set of generalized Bloch vectors is defined as the smallest dimension of affine spaces that contain this set. In this terminology, a set of density matrices is IC iff the (corresponding) set of generalized Bloch vectors has affine dimension $d^2-1$.  For a qubit, this fact leads to a simple geometric criterion. 
\begin{proposition}\label{pro:qubitIC}
	A set of qubit density matrices is IC iff the  set of Bloch vectors is not contained in any disk---the intersection of the Bloch ball with a plane.
\end{proposition}

In addition, any set  of qubit states that is contained in a disk in the Bloch ball is maskable \cite{ModiPSS18,LianLF19,DingH20} and thus not antiscrambling. 
To see this, it suffices to consider pure states. Up to a unitary transformation and irrelevant phase factors,  such a set of kets  is contained in a set of the form 
\begin{equation}
\scrK=\{\cos\theta|0\>+ \rme^{\rmi\alpha}\sin\theta |1\>\, |\, 0\leq \alpha<2\pi\},
\end{equation}
where $0\leq \theta\leq \pi/4$.
This set can be masked by the isometry:  $|0\>\mapsto|00\>, |1\>\mapsto |11\> $ \cite{ModiPSS18}. So \thsref{thm:HideIC} and \ref{thm:MaskIC}
imply the following corollary. 
\begin{corollary}\label{cor:MaskQubit}
	A set of qubit states is hidable (or maskable) iff the set is not IC, iff the set of Bloch vectors is contained in a disk.
\end{corollary}
\Crref{cor:MaskQubit} confirms the hyperdisk conjecture proposed in \rcite{ModiPSS18} for the qubit case. Here the proof is simpler than previous proofs and has wider applicability \cite{LianLF19,DingH20}. 
This corollary also shows that a nontrivial disk (not a single point) in the qubit case
is a maximal maskable set. In particular, $\scrD^\rmR(\caH)$ is a maximal maskable set, which is consistent with \thref{thm:MaskMax} presented in \sref{sec:MaskRQ} below. \Crref{cor:MaskQubit} also implies that any hidable or maskable set in the qubit case has measure zero. This conclusion is not a coincidence as confirmed by the following corollary.

\begin{corollary}
	Any hidable  set of pure (mixed) states in $\scrP(\caH)$ ($\scrD(\caH)$) has measure zero with respect to the uniform measure, and so does any maskable  set of pure (mixed) states.
\end{corollary}
Here the uniform measure means the measure induced by the Hilbert-Schmidt distance. In the case of pure states, it coincides with the measure induced by the Haar measure on the unitary group. For pure states, \rcite{LianLFF20}  proved that any maskable set has measure zero given the assumption that the output Hilbert spaces $\caH_A$ and $\caH_B$ have the same dimension as $\caH$. However, this assumption is too restrictive   according to our study on the masking of real quantum states
 as presented in \sref{sec:MaskRQ}.

\begin{proof}
	Suppose $\scrS\in \scrP(\caH)$ is hidable.  Then $\scrS$ is not IC according to \thref{thm:HideIC}. Therefore, we can find a nonzero hermitian operator $Q$ such that $\tr(Q\rho)=0$ for all $\rho\in \scrS$, which implies that $\scrS$ has measure zero within $\scrP(\caH)$. A similar proof applies to mixed states and to any hidable or maskable  set.
\end{proof}

Let $t$ be a positive integer. A weighted set of quantum states $\{|\psi_j\>, w_j\}_j$ in $\caH$  with $w_j> 0$ is a weighted $t$-design if $\sum_j w_j (|\psi_j\>\<\psi_j|)^{\otimes t}$ is proportional to the projector onto  the symmetric subspace in $\caH^{\otimes t}$; it is a $t$-design if all weights $w_j$ are equal \cite{Zaun11,ReneBSC04,Scot06}. Note that any weighted  $t$-design with $t\geq 2$ is automatically a weighted 2-design. In addition, the set $\{|\psi_j\>\}_j$ is IC if $\{|\psi_j\>, w_j\}_j$ is a  weighted 2-design for some choice of weights $w_j$ \cite{Scot06,ZhuE11}. Together with \thsref{thm:HideIC} and \ref{thm:MaskIC}, these observations yield the following corollary. 
\begin{corollary}\label{cor:design}
	Any  set of pure states that can  form a weighted 2-design is antiscrambling and not maskable. 	 
\end{corollary}
Prominent examples of $2$-designs include symmetric informationally complete (SIC) POVMs \cite{Zaun11,ReneBSC04,FuchHS17} and complete sets of $d+1$ mutually unbiased bases \cite{Ivan81,WootF89,DurtEBZ10}. Recall that a SIC~POVM in dimension $d$ is composed of $d^2$ pure states with the equal pairwise fidelity of $1/(d+1)$. Two bases $\{|\psi_j\>\}_j$ and $\{|\phi_k\>\}_k$ are mutually unbiased if $|\<\psi_j|\phi_k\>|^2=1/d$ for all $j,k$. 
A set of mutually unbiased bases in dimension $d$ is complete if it contains $d+1$ bases. 
According to \Crref{cor:design}, all these sets are  antiscrambling and not maskable.

\section{\label{sec:HR}Hurwitz-Radon  matrices}

In this section, we introduce the concept of  HR matrices  (operators) \cite{Hurw23,Rado22,Eckm06}, which is crucial to establishing our main result on the masking of real quantum states. Except for \lref{lem:HRtauTrace} below, most results presented in this section  are well known \cite{Eckm06}. 
For the convenience  of the reader, here we formulate these results in a language that is familiar to researchers in the quantum information community.

\subsection{Definition and  basic properties of HR matrices}

 A set $\{U_j\}_{j=1}^{s}$ of $m\times m$ unitary  matrices (operators)  is a set of HR matrices  (operators) \cite{Hurw23,Rado22,Eckm06} if 
\begin{equation}\label{eq:HR}
U_jU_k+U_kU_j=-2\delta_{jk}\openone.
\end{equation}
For example, let $\sigma_x, \sigma_y, \sigma_z$ be the three Pauli matrices; then 
$\{\rmi\sigma_x, \rmi\sigma_y, \rmi\sigma_z\}$ is a set of HR matrices in dimension~2.
Such matrices play  crucial roles in the Dirac equation and in the representations of the Clifford algebra \cite{Loun01book}.
The defining equation above  implies that
\begin{align}
U_j^2=-\openone,\quad U_j^\dag=-U_j,\quad \tr(U_j^\dag U_k)=m\delta_{jk}. 
\end{align}
So each  $U_j$ has at most two distinct eigenvalues, namely, $\rmi$ and $-\rmi$; in addition $U_j$ and $U_k$ are orthogonal with respect to the Hilbert-Schmidt inner product when $j\neq k$. If $s\geq 2$, then \eref{eq:HR}  implies that $\tr(U_j)=0$ (cf.~\lref{lem:HRtauTrace} below); in the case $s=1$, this conclusion does not hold automatically, and this fact has important implications for the masking of real density matrices, as we shall see later.

The main properties of HR matrices (operators) are summarized below and proved in \aref{sec:HRpropProof}.
\begin{restatable}{lemma}{lemHRequiDef}\label{lem:HRequiDef}
	Suppose $U_0, U_1, U_2, \ldots, U_s$ are $s+1$ unitary matrices of the same size and $\vec{c}=(c_0, c_1, \ldots, c_s)$ is a real vector of  $s+1$ components; let	
\begin{align}
|\vec{c}|:=\sqrt{\sum_{j=0}^s c_j^2},\quad   U(\vec{c}):=\sum_{j=0}^s c_j U_j. 
\end{align}	
 Then the following  statements are equivalent:
	\begin{enumerate}
		\item $\{U_0^\dag U_j\}_{j=1}^{s}$ is a set of $s$ HR matrices.
		\item $U(\vec{c})^\dag U(\vec{c})=|\vec{c}|^2 \openone$ for each real vector $\vec{c}$. 
		\item $U(\vec{c})$ is unitary for each normalized real vector $\vec{c}$. 
	\end{enumerate}
\end{restatable}

\begin{restatable}{lemma}{lemHRtauTrace}\label{lem:HRtauTrace}
	Suppose $U_0=\openone$ and $\{U_j\}_{j=1}^{s}$ with  $s\geq 2$ is a set of HR matrices that commute with a hermitian matrix $\tau$. Suppose $j,k, j',k'$ are integers that satisfy the conditions $0\leq j<k\leq s$ and $0\leq j'<k'\leq s$; let $\alpha>0$. Then 
	\begin{align}
	&\tr(U_j  U_k \tau^\alpha)=0, \label{eq:HRtauTrace1}\\
	&\tr(U_j  U_k U_{j'}  U_{k'} \tau^\alpha)=-\tr(\tau^\alpha)\delta_{j j'}\delta_{k k'},\quad s\neq 3, \label{eq:HRtauTrace2}  \\
	&\tr(U_j  U_k U_{j'}  U_{k'} \tau^\alpha)=-\tr(\tau^\alpha)\delta_{j j'}\delta_{k k'}\nonumber\\
	& +(\delta_{j=0}\epsilon_{kj'k'}+\delta_{j'=0}\epsilon_{j k k'})\tr(U_1U_2U_3\tau^\alpha),\;\, s=3, \label{eq:HRtauTrace3}
	\end{align}	
	where $\epsilon_{kj'k'}$ is equal to 1 ($-1$)  if 	$(k, j', k')$ is an even  (odd) permutation of $(1, 2, 3)$ and is equal to 0 otherwise; $\epsilon_{j k k'}$ is defined in a similar way.
\end{restatable}

When $\tau=\openone$, 
\lref{lem:HRtauTrace} yields
\begin{align}
&\tr(U_j  U_k)=0,\\
&\tr(U_j  U_k U_{j'}  U_{k'})=-m\delta_{j j'}\delta_{k k'},\quad s\neq 3,  \\
&\tr(U_j  U_k U_{j'}  U_{k'})=-m\delta_{j j'}\delta_{k k'}\nonumber\\
& +(\delta_{j=0}\epsilon_{kj'k'}+\delta_{j'=0}\epsilon_{j k k'})\tr(U_1U_2U_3),\quad s=3,
\end{align}
where $m=\tr(\openone)$.

\subsection{Construction of HR matrices}
In dimension 2, $\{\rmi\sigma_x, \rmi\sigma_y, \rmi\sigma_z\}$ is a set of HR matrices as mentioned before. In higher dimensions, HR matrices can be constructed iteratively. Given a set $\{U_j\}_{j=1}^{s}$ of $s$ HR matrices in dimension $m$, then  $s+2$ HR matrices can  be constructed in dimension $2m$ as follows,
\begin{equation}
\begin{gathered}
U_j\mapsto U_j\otimes \sigma_z, \quad j=1,2,\ldots, s,\\
U_{s+1}=\openone\otimes \rmi\sigma_x,\quad U_{s+2}=\openone\otimes \rmi\sigma_y. 
\end{gathered}
\end{equation}
In this way,  $2b+1$ HR matrices can be constructed in dimension $2^b$, which attain the maximum   number according to  \lref{lem:HRnumMax} below. If instead the dimension  is divisible by $2^b$, then the same number of HR matrices can be constructed by considering the tensor product with a suitable identity matrix.

Suppose $\{U_j\}_{j=1}^{r}$ and $\{V_j\}_{j=1}^{s}$ are two sets of HR matrices in dimensions $m_1$ and  $m_2$, respectively. Then  we can construct a set of $r+s+1$  HR matrices in dimension $2m_1m_2$ as follows \cite{Eckm06},
\begin{equation}\label{eq:HRcomposition}
W_j=\begin{cases}
U_j\otimes \openone\otimes \sigma_z & 1\leq j\leq r,\\
\openone\otimes V_{j-r}\otimes \sigma_x & r+1\leq j\leq r+s,\\
\openone\otimes\openone\otimes \rmi\sigma_y & j=r+s+1.
\end{cases}
\end{equation}
Note that this construction works even when  $r=0$ and $m_1=1$.

HR matrices can also be constructed from real orthogonal matrices. The maximum number of such HR matrices in dimensions 2, 4, 8 are 1, 3, 7, respectively. For example,  $\rmi\sigma_y$ is an HR matrix in  dimension 2, and  $\rmi\sigma_y\otimes\sigma_z$, $\rmi\sigma_y\otimes\sigma_x$, $\openone\otimes\rmi\sigma_y$ are three HR matrices  in dimension 4; seven HR matrices in dimension~8 can be constructed as follows, 
\begin{align}
&\sigma_x\otimes\rmi\sigma_y\otimes \openone,\, \sigma_x\otimes\sigma_z\otimes \rmi\sigma_y,\, \sigma_x\otimes\sigma_x\otimes \rmi\sigma_y,\, \rmi\sigma_y \otimes \openone\otimes\openone,\nonumber\\
&\sigma_z\otimes \openone\otimes\rmi\sigma_y,\,\sigma_z\otimes \rmi\sigma_y\otimes\sigma_z, \sigma_z\otimes \rmi\sigma_y\otimes\sigma_x. 
\end{align}
Real orthogonal HR matrices in higher dimensions can be constructed iteratively by virtue of \eref{eq:HRcomposition}, note that $W_j$ is real orthogonal if both $U_j$ and $V_{j-r}$ are real orthogonal. In particular $s+8$ real orthogonal HR matrices can be constructed in dimension $16m$ if $s$ real orthogonal HR matrices can be constructed in dimension $m$.

\begin{lemma}\label{lem:HRnumMax}
	$d-1$ unitary (real orthogonal)	HR matrices can be constructed in dimension $m$ iff $m$ is divisible by $\kappa(d)$ ($\kappa^\rmR(d)$), where
	\begin{align}
	\kappa(d)&:=2^{\lfloor(d-1)/2\rfloor}, \label{eq:MaskDimC}\\
	\kappa^\rmR(d)&:=\begin{cases}
	\kappa(d) & d=0,1,7 \mod 8,\\
	2\kappa(d) & d=2,3,4,5,6 \mod 8.
	\end{cases} \label{eq:MaskDimR}
	\end{align}
\end{lemma}
\Lref{lem:HRnumMax} follows from  \rscite{Hurw23,Rado22,Eckm06}. It clarifies the minimum dimension required to construct a given number of HR matrices. Incidentally, the variations of $\kappa(d)$ and   $\kappa^\rmR(d)$ exhibit Bott periodicity~\cite{Eckm06}, which is of key interest in algebraic topology.

\section{\label{sec:MaskRQ}Masking quantum information in real quantum mechanics}
In this section we  show that the set $\scrD^\rmR(\caH)$ of real density matrices is maskable and is actually a maximal maskable set. In addition, we  determine the entanglement of masking and the minimum dimension of the output Hilbert space required to mask $\scrD^\rmR(\caH)$. 

\subsection{Construction of a masker}
 Let $m$ be an even positive integer that is divisible by $\kappa(d)$.
Let $\caH_A$ and $\caH_B$ be two $m$-dimensional Hilbert spaces and 
\begin{align}
|\Phi\>:=\frac{1}{\sqrt{m}}\sum_{l=0}^{m-1} |ll\>
\end{align}
the canonical maximally entangled state in $\caH_A\otimes\caH_B$.
By assumption we can construct  a set $\{U_j\}_{j=1}^{d-1}$ of  $d-1$ traceless HR matrices  on $\caH_A$; let $U_0=\openone_A$.
Define the isometry $M: \caH\mapsto \caH_A\otimes \caH_B$ by its action,
\begin{equation}\label{eq:MaskerReal}
M|j\>= |\Phi_j\>:=(U_j \otimes \openone_B )|\Phi\>, \quad j=0,1,\ldots, d-1;
\end{equation}
then $M$ is a masker for $\scrP^\rmR(\caH)$ and $\scrD^\rmR(\caH)$. 

To verify the above claim, note that the $d$  states $|\Phi_j\>$ are orthonormal since
\begin{align}
\<\Phi_j|\Phi_k\>=\frac{1}{m}\tr(U_j^\dag U_k)=\delta_{jk}.
\end{align}
Given any normalized vector $\vec{c}=(c_0,c_1,\ldots, c_{d-1})$ and the state $|\psi(\vec{c})\>=\sum_j c_j|j\>\in \caH$, the output state has the form 
\begin{equation}
|\Psi(\vec{c})\>=\sum_j c_j|\Phi_j\>=[U(\vec{c}) \otimes \openone_B]|\Phi\>,\quad U(\vec{c})=\sum_j c_j U_j.
\end{equation}
If in addition $\vec{c}$ is a  real vector, then  $U(\vec{c})$ is unitary by \lref{lem:HRequiDef} given that $\{U_j\}_{j=1}^{d-1}$ is a set of HR matrices, which implies that $|\Psi(\vec{c})\>$ is maximally entangled. 
So $M$ is a masker for $\scrP^\rmR(\caH)$ and  $\scrD^\rmR(\caH)$.

If $m$ is divisible by $\kappa^\rmR(d)$ and each $U_j$ is real orthogonal, then the output Hilbert space could be real, which leads to a real masker for $\scrP^\rmR(\caH)$ and  $\scrD^\rmR(\caH)$. Therefore, the set of real states can be masked even with a real masker. In other words, the information about quantum states in real quantum mechanics can be completely hidden in the correlations. 
This fact shows  that the no-hiding and no-masking theorems do not apply to real quantum mechanics, in sharp contrast with complex quantum mechanics, as illustrated in \fref{fig:NoGo}.

When $d=4$ for example, a masker for $\scrD^\rmR(\caH)$ can be constructed using the set of HR matrices $\{\rmi\sigma_x, \rmi\sigma_y, \rmi\sigma_z\}$, which is tied to the two-qubit magic basis \cite{HillW97}
\begin{equation}\label{eq:MagicBasis}
\begin{aligned}
\!\!|\Phi_0\>&=\frac{1}{\sqrt{2}}(|00\>+|11\>), &
|\Phi_1\>&=\frac{\rmi}{\sqrt{2}}(|01\>+|10\>),\\
\!\!|\Phi_2\>&=\frac{1}{\sqrt{2}}(|01\>-|10\>),& |\Phi_3\>&=\frac{\rmi}{\sqrt{2}}(|00\>-|11\>).
\end{aligned}
\end{equation}
Let $M$ be the isometry from $\caH$ to the two-qubit Hilbert space $\caH_A\otimes\caH_B$ as defined by the map $|j\>\mapsto |\Phi_j\>$ for $j=0,1,2,3$. Then  $M$ can mask the set $\scrD^\rmR(\caH)$. Here the concurrence of the output state $|\Psi(\vec{c})\>=\sum_j c_j|\Phi_j\>$ reads \cite{HillW97}
\begin{equation}
C(|\Psi(\vec{c})\>)=\biggl|\sum_j c_j^2\biggr|,
\end{equation}
so $|\Psi(\vec{c})\>$ is maximally entangled iff $|\psi(\vec{c})\>$ is a real ket up to an overall phase factor. In this way, the masker~$M$ establishes a one-to-one correspondence between real pure states in $\caH$ and maximally entangled states in $\caH_A\otimes\caH_B$. This correspondence is tied to the geometric facts that normalized real kets in $\caH$ form a 3-sphere, and so does the special unitary group $\mathrm{SU}(2)$. Incidentally, based on this example, we can construct a hidable set that is not maskable,  as shown in \sref{sec:HideNoMask} later.

%
\subsection{All maskers for real quantum states}
Next, we determine all maskers for the sets $\scrP^\rmR(\caH)$ and $\scrD^\rmR(\caH)$. Let $|\Psi_0\>$ be any bipartite ket in $\caH_A\otimes \caH_B$ and $\tau_A=\tr_B(|\Psi_0\>\<\Psi_0|)$. Suppose $\{U_j\}_{j=1}^{d-1}$ is a set of HR matrices that commute with $\tau_A$ and satisfy $\tr(U_j\tau_A)=0$ for $j=1,2,\ldots, d-1$ [the condition $\tr(U_j\tau_A)=0$ is automatically guaranteed when $d\geq 3$ thanks to \lref{lem:HRtauTrace}]. Let
\begin{align}
|\Psi_j\>=(U_j \otimes \openone_B )|\Psi_0\>;
\end{align} 
then the kets $|\Psi_j\>$ for $j=0,1,\ldots, d-1$ are orthonormal, and the isometry $M: j\mapsto |\Psi_j\>$ is a masker for $\scrP^\rmR(\caH)$ and $\scrD^\rmR(\caH)$.
The following theorem  proved in \aref{sec:MaskGenProof} shows that essentially all maskers can be constructed in this way. It implies \crsref{cor:MaskSpecCh}-\ref{cor:PurityCon} below, which are proved in \aref{sec:MaskGenCorProof}. 
\begin{restatable}{theorem}{thmMaskGenHR}\label{thm:MaskGenHR}
	Suppose  the isometry $M: \caH\mapsto \caH_A\otimes \caH_B $ is a masker  for $\scrD^\rmR(\caH)$ with common reduced density matrices $\tau_A$ and $\tau_B$ of full rank. Let $|\Psi_j\>:=M|j\>$. 
	Then $|\Psi_j\>$ for $j=1,2,\ldots,d-1$ have the form
	\begin{equation}\label{eq:PsijMaskGen}
	|\Psi_j\>=(U_j \otimes \openone_B )|\Psi_0\>=(\openone_A \otimes V_j )|\Psi_0\>, 
	\end{equation}		
	where $\{U_j\}_{j=1}^{d-1}$ and $\{V_j\}_{j=1}^{d-1}$
	are two  sets of HR matrices. Meanwhile,
	$U_j$ and $ V_j$ commute with $\tau_A$ and $\tau_B$, respectively,  and satisfy $\tr(U_j\tau_A)=\tr(V_j\tau_B)=0$.
\end{restatable}

\begin{restatable}{corollary}{corMaskSpecCh}\label{cor:MaskSpecCh}
	When $d\geq 3$,	there exists a complex (real) masker for $\scrD^\rmR(\caH)$ with masking spectrum $\{\lambda_o, m_o\}_o$ iff the multiplicity $m_o$ of each nonzero eigenvalue $\lambda_o$ is divisible by $\kappa(d)$ ($\kappa^\rmR(d)$). When $d=2$, 
	the same conclusion holds for a real masker; 
	there exists a complex masker for $\scrD^\rmR(\caH)$ with masking spectrum $\{\mu_l\}_l$ iff there exists a vector $\vec{v}=(v_l)_l$ composed of 1 and $-1$ such that $\sum_l v_l\mu_l=0$.
\end{restatable}

Here $\lambda_o$    refer to distinct eigenvalues, while  $\mu_l$ for different $l$ may equal each other.

\begin{restatable}{corollary}{corMaskDim}\label{cor:MaskDim}
	$\scrD^\rmR(\caH)$ can be masked in the  Hilbert space $\caH_A\otimes \caH_B$ iff $d_A, d_B\geq \tilde{\kappa}(d):= \max\{\kappa(d),2\}$; if the Hilbert space $\caH_A\otimes \caH_B$ is real, then the condition turns into $d_A, d_B\geq \kappa^\rmR(d)$.
\end{restatable}

Given a density matrix $\rho\in \scrD(\caH)$, the robustness of imaginarity $\scrI_\rmR(\rho)$ \cite{HickG18,WuKRS21} is defined as 
\begin{equation}\label{eq:RoI}
\scrI_\rmR(\rho):=\min \biggl\{a\geq 0\Big| \rho' \in \scrD(\caH), \frac{\rho+a \rho'}{1+a}\in \scrD^\rmR(\caH) \biggr\}. 
\end{equation}
A closed formula was presented in  \rcite{WuKRS21} as reproduced here,
	\begin{align}\label{eq:RoIvalue}
\scrI_\rmR(\rho)=\frac{1}{2}\|\rho-\rho^\rmT\|_1, 
\end{align}
where $\|\rho-\rho^\rmT\|_1=\tr \sqrt{(\rho-\rho^\rmT)^2}$ is the Schatten 1-norm of $\rho-\rho^\rmT$. 
Here we are interested in $\scrI_\rmR(\rho)$ because of its intimate connection with 
the properties of the output state of any masker for $\scrD^\rmR(\caH)$ with $d=\dim(\caH)\geq 3$ as shown in the following corollary; see also \eref{eq:ConcurrenceRoI} below.  It should be noted that this corollary is not applicable when $d=2$. 
\begin{restatable}{corollary}{corPurityCon}\label{cor:PurityCon}
	Suppose  $d=\dim(\caH)\geq 3$ and $M$ is a masker for the set $\scrD^\rmR(\caH)$ with common reduced states $\tau_A$ and $\tau_B$. Then $\tau_A$ ($\tau_B$) is majorized by $\varrho_A$ ($\varrho_B$)  defined in \eref{eq:varrhoAB} for all $\rho\in \scrD(\caH)$. In addition,
	\begin{gather}
	\tr(\varrho_A^2)+	\tr(\varrho_B^2)=\wp\bigr(2+\|\rho-\rho^\rmT\|_{\hs}^2\bigr), \label{eq:purityABsum}\\
	2\wp\Bigl[1+\frac{2}{d}\scrI_\rmR(\rho)^2\Bigr]\leq 	\tr(\varrho_A^2)+	\tr(\varrho_B^2)\leq 2\wp\bigl[1+\scrI_\rmR(\rho)^2\bigr],	\label{eq:purityABsumLBUB}	
	\end{gather}
	where $\|\rho-\rho^\rmT\|_{\hs}=\sqrt{\tr[(\rho-\rho^\rmT)^2]}$ is the Hilbert-Schmidt norm of $\rho-\rho^\rmT$, and
 $\wp$ is the masking purity  defined in \eref{eq:MaskPurity}.
	If in addition $d\neq 4$, then 	
	\begin{gather}
	\tr(\varrho_B^2)=	\tr(\varrho_A^2)=\frac{1}{2}\wp\bigr(2+\|\rho-\rho^\rmT\|_{\hs}^2\bigr), \label{eq:purityAB}\\
	\wp\Bigl[1+\frac{2}{d}\scrI_\rmR(\rho)^2\Bigr]\leq 	\tr(\varrho_A^2)\leq \wp\bigl[1+\scrI_\rmR(\rho)^2\bigr].	\label{eq:purityABLBUB}
	\end{gather}	
\end{restatable}
When $d\geq 3$, \eref{eq:purityABsum} implies that
\begin{align}
\tr(\varrho_A^2)+\tr(\varrho_B^2)\geq 2\wp,
\end{align} 
and the inequality is saturated iff $\rho$ is a real density matrix, that is, $\rho\in \scrD^\rmR(\caH)$. If in addition $d\neq 4$, then \eref{eq:purityAB} implies that 
\begin{equation}
\tr(\varrho_A^2)\geq \wp,\quad \tr(\varrho_B^2)\geq \wp,
\end{equation}
and each inequality is saturated iff $\rho$ is a real density matrix.

\subsection{Entanglement of masking and masking purity}
By virtue of \Thref{thm:MaskGenHR} and \crref{cor:MaskSpecCh}, we can determine 
the minimum entanglement cost required to mask the set of real quantum states. The following corollary is proved in \aref{sec:MaskEntProof}.
\begin{restatable}{corollary}{corMaskEnt}\label{cor:MaskEnt}
	Let $E$ be an entanglement monotone; then 
	\begin{align}
	E(\scrD^\rmR(\caH))&=E(\scrP^\rmR(\caH))=E(|\Phi(\tilde{\kappa}(d))\>), \label{eq:MaskEntR}\\
	E^\rmR(\scrD^\rmR(\caH))&=E^\rmR(\scrP^\rmR(\caH))=E(|\Phi(\kappa^\rmR(d))\>),\label{eq:MaskEntRR}
	\end{align}
	where $|\Phi(m)\>$ is a maximally entangled state of Schmidt rank $m$. 
\end{restatable}

\begin{figure}
	\includegraphics[width=7cm]{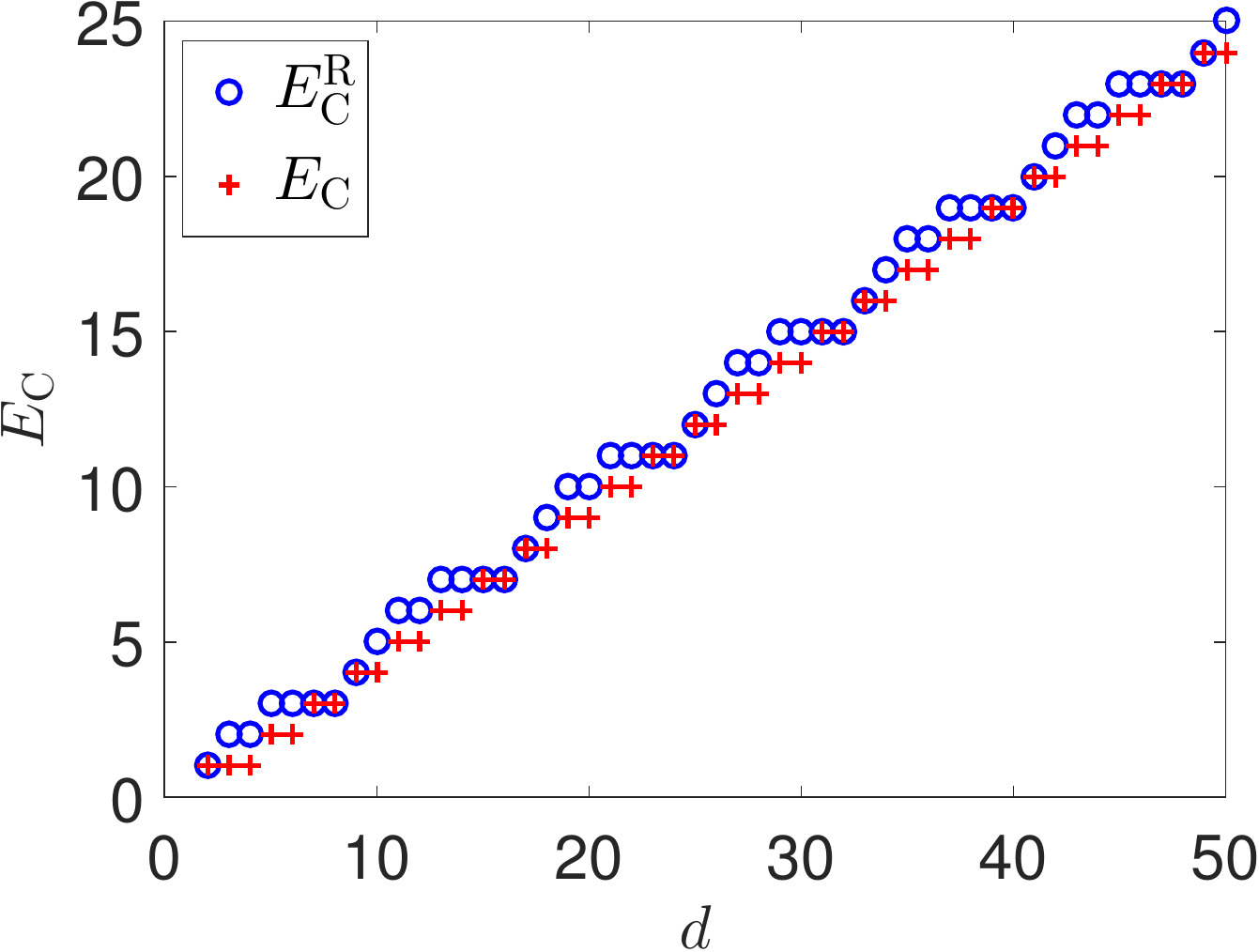}
	\caption{\label{fig:EntMask}
		Minimum entanglement cost $E_\rmC$ ($E_\rmC^\rmR$) required to mask $\scrP^\rmR(\caH)$ and $\scrD^\rmR(\caH)$ with complex (real)  Hilbert spaces. }
\end{figure}

For a bipartite pure state, the entanglement cost, entanglement of formation, distillable entanglement, and relative entropy of entanglement all coincide with the von Neumann entropy of each reduced state \cite{HoroHHH09}. If $E$ is one of these entanglement measures, say,  the entanglement cost $E_\rmC$, then \crref{cor:MaskEnt} and \lref{lem:HRnumMax} imply that
\begin{align}
&E_\rmC(\scrD^\rmR(\caH))=E_\rmC(\scrP^\rmR(\caH))=\log_2\tilde{\kappa}(d) \nonumber\\
&=\begin{cases}
1 &d=2,\\
\bigl\lfloor \frac{d-1}{2}\bigr\rfloor & d\geq 3;
\end{cases} \label{eq:MaskEntR2}\\
&E_\rmC^\rmR(\scrD^\rmR(\caH))=E_\rmC^\rmR(\scrP^\rmR(\caH))=\log_2\kappa^\rmR(d)\nonumber\\
&=\begin{cases}
\bigl\lfloor \frac{d-1}{2}\bigr\rfloor & d=0,1,7 \mod 8,\\
\bigl\lfloor \frac{d-1}{2}\bigr\rfloor+1 & d=2,3,4,5,6 \mod 8.
\end{cases}
\label{eq:MaskEntRR2}
\end{align}
Note  that $E_\rmC(\scrD^\rmR(\caH))\approx E_\rmC^\rmR(\scrD^\rmR(\caH))\approx d/2$. 
The entanglement of masking quantified by the entanglement cost $E_\rmC$ is approximately linear in $d$, and the deviation exhibits Bott periodicity, as  illustrated in \fref{fig:EntMask}

If the concurrence $C$ \cite{HillW97,Woot98,RungBCH01} is chosen to quantify the entanglement, then we have 
\begin{align}
&C(\scrD^\rmR(\caH))=C(\scrP^\rmR(\caH))=\sqrt{2[1-\tilde{\kappa}(d)^{-1}]}, \\
&C(\scrD^\rmR(\caH))=C(\scrP^\rmR(\caH))=
\sqrt{2[1-\kappa^\rmR(d)^{-1}]}.
\end{align}
In addition, \crref{cor:MaskSpecCh} implies that the masking purity $\wp$ of any masker for $\scrD^\rmR(\caH)$ satisfies $\wp\leq 1/\tilde{\kappa}(d)$ (cf. the proof of \crref{cor:MaskEnt}); by contrast,  $\wp\leq 1/\kappa^\rmR(d)$ for any real masker.
This conclusion is expected given the close connection between the concurrence of a bipartite pure state and the purity of each reduced state.

\subsection{Real quantum mechanics as a maximal maskable subtheory}
Here we show that the set $\scrD^\rmR(\caH)$ of real density matrices is  a maximal maskable set within  $\scrD(\caH)$. Such maximal maskable sets are valuable to understanding the potential and limit of hiding and masking quantum information. However, no maximal maskable sets beyond the qubit system have been found before the current study as far as we know. This conclusion shows that real quantum mechanics is a maximal maskable subtheory in the usual quantum mechanics.
\Thsref{thm:MaskMax} and \ref{thm:MaskEnt} below are proved in \aref{sec:MaskMaxProof}.  

\begin{restatable}{theorem}{thmMaskMax}\label{thm:MaskMax}
	$\scrD^\rmR(\caH)$ is a maximal maskable set. 
\end{restatable}
\Thref{thm:MaskMax} implies that $\scrP^\rmR(\caH)$ is a maximal maskable set within $\scrP(\caH)$. Moreover, given any masker $M$ for $\scrD(\caH)$ with $d=\dim(\caH)\geq 3$, it turns out $M(\scrP^\rmR(\caH))$ happens to be the set of maximally entangled states within $M(\scrD(\caH))$, as shown in the following theorem. 

\begin{restatable}{theorem}{thmMaskEnt}\label{thm:MaskEnt}
	Suppose  $d\geq 3$ and   $M: \caH\mapsto \caH_A\otimes \caH_B$ is a masker  for $\scrD^\rmR(\caH)$.
	Let $|\Psi_0\>=M|0\>$; then any state in $M(\scrD(\caH))$ can be created from $|\Psi_0\>\<\Psi_0|$ by local operations and classical communication (LOCC). 
	In addition, $E(M(\rho))\leq E(|\Psi_0\>\<\Psi_0|)$ for any entanglement monotone $E$ and  all $\rho\in \scrD(\caH)$. When $E$ is a convex strict monotone, the upper bound is saturated iff $\rho\in \scrP^\rmR(\caH)$.
\end{restatable}

A strict entanglement monotone $E$ means $E(|\Psi_1\>)>E(|\Psi_2\>)$ whenever $|\Psi_1\>$ can be turned into $|\Psi_2\>$ by LOCC, but not vice versa. Prominent examples of convex strict  monotones include entanglement of formation, entanglement cost, relative entropy of entanglement, and concurrence  \cite{HillW97,Woot98,RungBCH01,HoroHHH09}. For concurrence,  we can further derive the following result (assuming $d\geq 3$)
\begin{align}
C(\varrho)&\leq \min\Bigl\{\sqrt{2[1-\tr(\varrho_A^2)]},\sqrt{2[1-\tr(\varrho_B^2)]}\,\Bigr\}\nonumber\\
&\leq \sqrt{2-\wp\bigr(2+\|\rho-\rho^\rmT\|_{\hs}^2\bigr)}\leq \sqrt{2(1-\wp)}, \label{eq:concurrence}
\end{align}
where $\varrho_A=\tr_B(\varrho)$ and $\varrho_B=\tr_A(\varrho)$ with $\varrho=M(\rho)$. 
Here the second inequality follows from \eref{eq:purityABsum} in \crref{cor:PurityCon}. 
	If $\rho$ is pure, then $\varrho$ is pure, and   $\varrho_A, \varrho_B$ have the same purity, so the first two inequalities in \eref{eq:concurrence} are saturated, which yields
\begin{align}
C(\varrho)&= \sqrt{2-\wp\bigr(2+\|\rho-\rho^\rmT\|_{\hs}^2\bigr)}\leq \sqrt{2(1-\wp)} \label{eq:concurrencePure}
\end{align}
for  $d\geq 3$.	
By contrast,  the last inequality in \eref{eq:concurrence} [also in \eref{eq:concurrencePure}] is saturated iff $\rho$ is a real density matrix.

\begin{figure}
	\includegraphics[width=7cm]{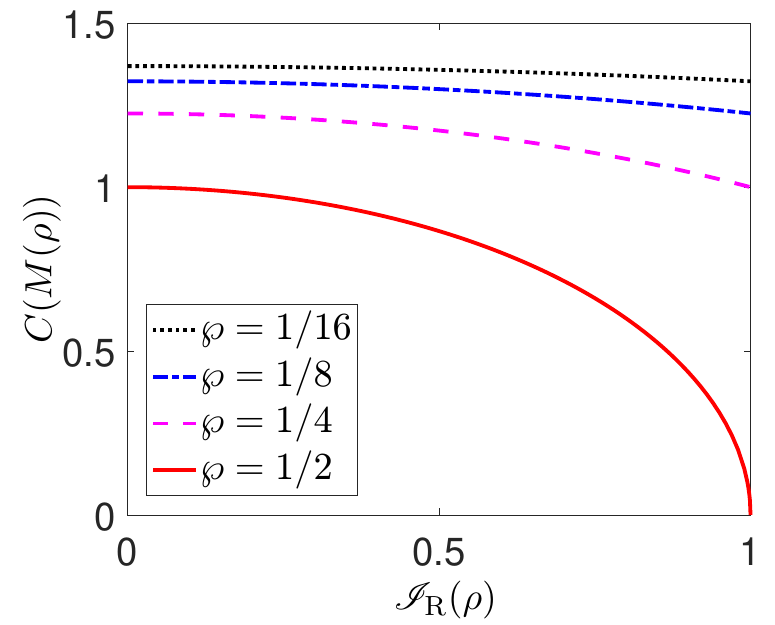}
	\caption{\label{fig:MaskCon}
		The relation between the concurrence $C(M(\rho))$ of the output state of the masker $M$ for $\scrD^\rmR(\caH)$ and the robustness of imaginarity $\scrI_\rmR(\rho)$ of the input state. Here 	$\rho$ is any pure state on $\caH$ with $\dim(\caH)\geq 3$, and $\wp$ is the masking purity of $M$. The concurrence decreases monotonically with the robustness of imaginarity. 	
	}
\end{figure}	
	
	For a pure input state $\rho$, the concurrence $C(M(\rho))$ of the output state is tied to the robustness of imaginarity $\scrI_\rmR(\rho)$ defined in \eref{eq:RoI} \cite{HickG18,WuKRS21}, which plays a key role in the resource theory of imaginarity. In particular, \esref{eq:RoIvalue}{eq:concurrencePure} together imply that 
	\begin{align}\label{eq:ConcurrenceRoI}
C(M(\rho))&= \sqrt{2-2\wp-2\wp\scrI_\rmR^2(\rho)},\quad   d\geq 3,
\end{align}
note that $\|\rho-\rho^\rmT\|_1=\sqrt{2}\|\rho-\rho^\rmT\|_{\hs}$ and 
\begin{align}\label{eq:RoI2}
\scrI_\rmR(\rho)=\frac{1}{2}\|\rho-\rho^\rmT\|_1=\frac{1}{\sqrt{2}}\|\rho-\rho^\rmT\|_{\hs}
\end{align}	
 when $\rho$ is pure (which means $\rho^\rmT$ is also pure). In this case $C(M(\rho))$  is completely determined by the masking purity $\wp$ and robustness of imaginarity $\scrI_\rmR(\rho)$.  This result is illustrated in \fref{fig:MaskCon} and is of  intrinsic interest to the resource theory of imaginarity.

Here it should be pointed out that \thref{thm:MaskEnt} and Eqs.~\eqref{eq:concurrence}-\eqref{eq:ConcurrenceRoI} are not applicable when $d=2$ since \crref{cor:PurityCon} is not applicable in this special case. To construct a counterexample, let 
\begin{equation}
\begin{aligned}
|\Psi_0\>&=\frac{1}{2}|00\>+\frac{1}{2}|11\>+\frac{1}{\sqrt{2}}|22\>,\\
|\Psi_1\>&=\frac{\rmi}{2}|00\>+\frac{\rmi}{2}|11\>-\frac{\rmi}{\sqrt{2}}|22\>.
\end{aligned}
\end{equation}
Then $|\Psi_0\>, |\Psi_1\>$ are orthonormal,  and the isometry $M$ defined by the map	$|j\>\mapsto |\Psi_j\>$ for $j=0, 1$ is a masker for $\scrD^\rmR(\caH)$ with masking purity $\wp=3/8$ [cf. \eref{eq:Maskerd2} in \aref{sec:MaskGenCorProof}]. Let 
\begin{align}
|\psi\>=\frac{1}{6}\Bigl[\sqrt{6(3-2\sqrt{2}\lsp)}\lsp|0\>-\rmi\sqrt{6(3+2\sqrt{2}\lsp)}\lsp|1\>\Bigr];
\end{align}
then 
\begin{align}
M|\psi\>&=\frac{1}{6}\Bigl[\sqrt{6(3-2\sqrt{2}\lsp)}\lsp|\Psi_0\>-\rmi\sqrt{6(3+2\sqrt{2}\lsp)}\lsp|\Psi_1\>\Bigr]\nonumber\\
&=\frac{1}{\sqrt{3}}(|00\>+|11\>-|22\>).
\end{align}
According to the majorization criterion \cite{Niel99}, $M|\psi\>$  can be turned into  $|\Psi_0\>$ by LOCC, but not vice versa. Therefore, $E(M(|\psi\>\<\psi|))\geq E(|\Psi_0\>\<\Psi_0|)$ for any entanglement monotone $E$, and the inequality is strict when $E$ is the entanglement cost, entanglement of formation, distillable entanglement,  relative entropy of entanglement, or concurrence, which contradicts \thref{thm:MaskEnt}. For example, the concurrence of $M(|\psi\>\<\psi|)$  reads
\begin{equation}
C(M(|\psi\>\<\psi|))=\frac{2}{\sqrt{3}}>\sqrt{2(1-\wp)}=\frac{\sqrt{5}}{2},
\end{equation}
which contradicts Eqs.~\eqref{eq:concurrence}-\eqref{eq:ConcurrenceRoI}.

\section{\label{sec:HideNoMask}A hidable set that is not maskable}
In this section we provide an example which is hidable but not maskable. Such an example has never been constructed in the literature previously
as far as we know.  Suppose $\dim(\caH)=4$ and 
$\scrS\subset \scrD(\caH)$ is the set of density matrices $\rho=\sum_{j,k=0}^3\rho_{jk}|j\>\<k|$ defined by the  following constraints
\begin{align}
\Im \rho_{01}=\Im \rho_{23},\quad \Im \rho_{02}=-\Im \rho_{13},\quad  \Im \rho_{03}=\Im \rho_{12},
\end{align}
where $\Im$ denotes the imaginary part of a complex number. Then $\scrS$ contains $\scrD^\rmR(\caH)$  as a strict subset and is thus not maskable according to \thref{thm:MaskMax}. Nevertheless, $\scrS$ 
is hidable as we shall see shortly. 

To show that $\scrS$ 
is hidable, it suffices to construct a nontrivial partial masker for $\scrS$. 
Consider the masker $M$ for $\scrD^\rmR(\caH)$  that is constructed using the magic basis in \eref{eq:MagicBasis}. Although $M$ is not a masker for $\scrS$, it is straightforward to verify that 
\begin{align}
\tr_B[M(\rho)]=\frac{\openone_A}{2}\quad \forall \rho\in \scrS.
\end{align}
Therefore, $M$ is a partial masker for $\scrS$. This partial masker is nontrivial because the information about all real quantum  states  in $\scrS$ is encoded in the correlations. This example demonstrates that not all hidable sets of quantum states are maskable. In other words, not all sets of quantum states that are not maskable are antiscrambling; the property of being antiscrambling is much stronger, as illustrated in \fref{fig:NoGo}.

\section{\label{sec:PhaseReal}Relation between real states and phase-parameterized states}
Let $\vec{c}=(c_0,c_1, \ldots, c_{d-1})$ be a normalized real vector with $d\geq 2$ and $c_j>0$ for $j=0,1,\ldots,d-1$; define
\begin{equation}
\mathscr{K}(\vec{c})=\biggl\{\sum_j c_j \rme^{\rmi\phi_j} |j\rangle\bigg| 0\leq \phi_j< 2\pi, j=0,1,\ldots,d-1 \biggr\}.
\end{equation}
Let $\scrP(\vec{c})$ be the set of phase-parameterized states associated with kets in $\scrK(\vec{c})$ and let $\overline{\scrP}(\vec{c})$ be the convex hull of $\scrP(\vec{c})$. According to \rcite{ModiPSS18}, the sets $\scrP(\vec{c})$ and $\overline{\scrP}(\vec{c})$ can be masked by the isometry defined as follows,
\begin{equation}
|j\rangle\mapsto |jj\rangle, \quad j=0,1,\ldots d-1.
\end{equation}
This isometry is well known although it was not recognized as a masker before the study in  \rcite{ModiPSS18}. Here the dimension of each subsystem of the output system is the same as that of the input system, which is much smaller compared with the counterpart for masking the set of real quantum states as discussed in \sref{sec:MaskRQ}. 
Moreover, \rcite{ModiPSS18} conjectured that any maskable set can be embedded in $\overline{\scrP}(\vec{c})$ for a suitable choice of $\vec{c}$, which is known as the hyperdisk conjecture. Although this conjecture has been disproved \cite{DingH20}, it is plausible that $\overline{\scrP}(\vec{c})$ is a maximal maskable set.

\begin{conjecture}\label{con:PhaseMax}
Suppose $c_j>0$ for $j=0,1,\ldots,d-1$; then $\overline{\scrP}(\vec{c})$ is a maximal maskable set. 
\end{conjecture}

Here we clarify the relation between the set $\scrP(\vec{c})$ of phase-parameterized states defined above and 
the set $\scrP^\rmR(\caH')$ of real states. For generality, here $\caH'$ is a Hilbert space whose dimension may be different from that of $\caH$. 
When $d=2$ and $c_0=c_1=1/\sqrt{2}$, the set $\scrP(\vec{c})$ corresponds to a great circle on the Bloch sphere and is equivalent to the set $\scrP^\rmR(\caH')$  with $\dim(\caH')=2$. Here "equivalent" means that $\scrP(\vec{c})$ can be mapped to $\scrP^\rmR(\caH')$ by some isometry $S$, that is
\begin{equation}
S\scrP(\vec{c})S^\dag =\scrP^\rmR(\caH'). 
\end{equation}
It turns out this is the only case in which $\scrP(\vec{c})$ is equivalent to a subset of $\scrP^\rmR(\caH')$, as shown in \pref{pro:PhaseRoReal} below and proved in \aref{sec:PhaseRealProof}. Conversely, $\scrP^\rmR(\caH')$ is equivalent to a subset of $\scrP(\vec{c})$ only if $\dim(\caH')=2$,
as shown in  \pref{pro:RealToPhase} below and proved in \aref{sec:PhaseRealProof}.

\begin{restatable}{proposition}{proPhaseRoReal}\label{pro:PhaseRoReal}
	Suppose $\scrP(\vec{c})$ with $d\geq 2$ is equivalent to a subset of $\scrP^\rmR(\caH')$; then $d=2$ and $c_0=c_1=1/\sqrt{2}$. 
\end{restatable}

\begin{restatable}{proposition}{proRealToPhase}\label{pro:RealToPhase}
	Suppose $\scrP^\rmR(\caH')$ with $\dim(\caH')\geq 2$ is equivalent to a subset of $\scrP(\vec{c})$; then $\dim(\caH')=2$. 
\end{restatable}
\Psref{pro:PhaseRoReal} and \ref{pro:RealToPhase} still hold if $\scrP(\vec{c})$ is replaced by $\overline{\scrP}(\vec{c})$ and $\scrP^\rmR(\caH')$ is replaced by $\scrD^\rmR(\caH')$. By \pref{pro:RealToPhase},  the set $\scrD^\rmR(\caH')$ with $\dim(\caH')\geq3$ cannot be embedded in $\overline{\scrP}(\vec{c})$ irrespective of the choice of $\vec{c}$, which disproves again the hyperdisk conjecture posed in \rcite{ModiPSS18} (cf. \rcite{DingH20}). 
This conclusion is consistent with the fact that $\scrD^\rmR(\caH')$ is a maximal maskable set as shown in \thref{thm:MaskMax}. On the other hand, except for the special case specified in \pref{pro:PhaseRoReal}, $\overline{\scrP}(\vec{c})$ cannot be embedded in  $\scrD^\rmR(\caH')$, so  any maximal maskable set  containing $\overline{\scrP}(\vec{c})$ cannot be  equivalent to $\scrD^\rmR(\caH')$. Quite likely, $\overline{\scrP}(\vec{c})$ is itself a maximal maskable set as stated in \cref{con:PhaseMax}. In any case, \psref{pro:PhaseRoReal} and \ref{pro:RealToPhase}
 imply that there exist inequivalent maximal maskable sets, which reflects the complexity of the masking problem.

\section{\label{sec:sum}Summary} 
We showed that it is impossible to hide or mask any set of quantum states that is IC, which strengthens the  no-hiding and no-masking theorems and establishes an information theoretical underpinning of these no-go results.   In sharp contrast,  quantum information can be completely hidden in correlations for real quantum mechanics, although the minimum dimension of the output Hilbert space has to increase exponentially with the dimension of the original Hilbert space. Moreover, in a precise sense real quantum mechanics is a maximal maskable subtheory within  complex quantum mechanics.

Our study offers valuable insight on the potential and limit of hiding and masking quantum information, which are of intrinsic interest to foundational studies. 
Meanwhile, the masking protocols we constructed are  useful to quantum secret sharing and quantum communication. This is the case in particular when  quantum states with real density matrices are easier to prepare and manipulate, say, in linear optics \cite{WuKRS21}. 
Furthermore, our study may shed light on a number of other active research areas, including information scrambling and the black-hole information paradox. Notably, the no-hiding theorem derived in \rcite{BrauP07} aggravates the tension between unitarity and Hawking's
semiclassical analysis of  black-hole radiation. The tension would  not be so serious in real quantum mechanics thanks to the breakdown of the no-hiding theorem.
In addition, our study indicates that Bott periodicity might play an important role in these research areas. We hope that our work can stimulate further progresses in these research topics.

\bigskip
\section*{Acknowledgments}
This work is  supported by   the National Natural Science Foundation of China (Grant No.~11875110) and  Shanghai Municipal Science and Technology Major Project (Grant No.~2019SHZDZX01).

\appendix

\bigskip

\section*{Appendix}

In this Appendix, we prove the key results presented in the main text, including \lsref{lem:MaskEnt}-\ref{lem:HRtauTrace}, \thsref{thm:HideIC}-\ref{thm:MaskEnt}, \crsref{cor:MaskSpecCh}-\ref{cor:MaskEnt}, and \psref{pro:PhaseRoReal}, \ref{pro:RealToPhase}. For the convenience of the readers, lemmas, theorems,  corollaries, and propositions will be restated before their proofs.  In the course of study, we also derive a few auxiliary results on bipartite pure states and on the masking of real quantum states, which are of some independent interest.

\section{\label{sec:MaskEntLemProof}Proof of \lref{lem:MaskEnt}}
\lemMaskEnt*

\begin{proof}[Proof of \lref{lem:MaskEnt}]
	Suppose the masker $M$ is an isometry from $\caH$ to $\caH_A\otimes\caH_B$. 
	Let $\varrho_0=M(\rho_0)$ and  $\tau_A=\tr_B(\varrho_0)$. Let $\rho$ be an arbitrary quantum state in $\scrS$ and $\varrho=M(\rho)$; then $\tr_B(\varrho)=\tau_A$ given that $M$ is a masker for $\scrS$. Let $|\Psi_\varrho\>$ and $|\Psi_{\varrho_0}\>$ be  purifications of $\varrho$ and $\varrho_0$, respectively, in $\caH_A\otimes \caH_B\otimes \caH_C$, where $\caH_C$ is a suitable Hilbert space. Then 
	\begin{align}
	&\tr_{C}(|\Psi_\varrho\>\<\Psi_\varrho|)=\varrho,\quad \tr_{C}(|\Psi_{\varrho_0}\>\<\Psi_{\varrho_0}|)=\varrho_0;\\
	&\tr_{BC}(|\Psi_\varrho\>\<\Psi_\varrho|)=\tr_{BC}(|\Psi_{\varrho_0}\>\<\Psi_{\varrho_0}|)=\tau_A\quad \forall \varrho\in M(\scrS).
	\end{align}
	The last equation implies that 
	\begin{equation}\label{eq:MaskEntProof}
	E(\varrho_0)=E(|\Psi_{\varrho_0}\>\<\Psi_{\varrho_0}|)= E(|\Psi_\varrho\>\<\Psi_\varrho|)\geq E(\varrho)
	\end{equation}
	for all $\varrho\in M(\scrS)$,
	where $E(|\Psi_{\varrho_0}\>\<\Psi_{\varrho_0}|)$ and $E(|\Psi_\varrho\>\<\Psi_\varrho|)$ refer to the bipartite entanglement of $|\Psi_{\varrho_0}\>\<\Psi_{\varrho_0}|$ and $|\Psi_\varrho\>\<\Psi_\varrho|$ with respect to the partition between $A$ and $BC$. The two equalities in \eref{eq:MaskEntProof} follow from the fact that $\varrho_0$, $|\Psi_{\varrho_0}\>\<\Psi_{\varrho_0}|$, and $|\Psi_\varrho\>\<\Psi_\varrho|$  have the same nonzero Schmidt coefficients, while the inequality follows from the monotonicity of $E$ under partial trace. As an immediate corollary of \eref{eq:MaskEntProof}, we can deduce that
	\begin{equation}
	E(\scrS,M)=\max_{\rho\in \scrS} E(M(\rho))=E(\varrho_0)=E(M(\rho_0)),
	\end{equation}
	which confirms \lref{lem:MaskEnt}.
\end{proof}

\section{\label{sec:NoHMproof}Proofs of \thsref{thm:HideIC} and \ref{thm:MaskIC}}

\thmHideIC*
\begin{proof}[Proof of \thref{thm:HideIC}]
	Thanks to \pref{pro:MaskSpan}, it suffices to consider an IC set that is composed of pure states, in which case it is more convenient to work with kets. 	Let $\scrK=\{|\psi_h\>\}_h\subset \caH$ be an arbitrary set of kets  that is IC (which means the corresponding set of states is IC),  then $|\scrK|\geq d^2$ with $d=\dim(\caH)$.
	Without loss of generality, we can assume that  $d$ of the kets in $\scrK$, say $|\psi_0\>,|\psi_1\>, \ldots, |\psi_{d-1}\>$,  are linearly independent, so they form a basis (not necessarily orthogonal)  for $\caH$, denoted by $\caB:=\{|\psi_0\>,|\psi_1\>,\ldots, |\psi_{d-1}\>\}$ henceforth. Let $\tilde{\caB}=\{|\tilde{\phi}_0\>,|\tilde{\phi}_1\>,\ldots, |\tilde{\phi}_{d-1}\>\}$ be the dual basis of $\caB$ defined by the requirement 
	$\<\tilde{\phi}_j|\psi_k\>=\delta_{jk}$, where the tilde is used to indicate that the kets $|\tilde{\phi}_j\>$ are not necessarily normalized. Here and in the rest of this proof we assume that $j,k$ can take on the values $0,1,\ldots, d-1$.

	Suppose $M$ is an isometry from $\caH$ to $\caH_A\otimes \caH_B$ and  a partial masker for the set $\scrK$. Let $|\Psi_h\>=M|\psi_h\>$	and $\tau_A=\tr_B(|\Psi_h\>\<\Psi_h|)$; note that $\tau_A$ is independent of $h$ by assumption. Suppose $\tau_A$ has the  spectral decomposition $\tau_A=\sum \lambda_u |a_u\>\<a_u|$,
	where $\lambda_u$ are eigenvalues of $\tau_A$ and $|a_u\>$ form an orthonormal eigenbasis. Then  $|\Psi_h\>$ can be expressed as follows, 
	\begin{equation}
	|\Psi_h\>=\sum_u \sqrt{\lambda_u} |a_u\>|b_u^{(h)}\>,
	\end{equation}
	where the kets $|b_u^{(h)}\>$ for a given $h$ are orthonormal. 	In addition, $\sum_j |\psi_j\>\<\tilde{\phi}_j|$ is  the identity operator on $\caH$, so 
	\begin{equation}
	|\psi_h\>=\sum_j |\psi_j\>\<\tilde{\phi}_j|\psi_h\>=\sum_j\mu_{jh} |\psi_j\>,
	\end{equation}
	where $\mu_{jh}=\<\tilde{\phi}_j|\psi_h\>$ and the two summations run over $j=0,1,\ldots, d-1$. Accordingly, $|b_u^{(h)}\>=\sum_j \mu_{jh} |b_u^{(j)}\>$, and the requirement $\<b_v^{(h)}|b_u^{(h)}\>=\delta_{uv}\<\psi_h|\psi_h\>$ implies that
	\begin{align}
	\<\psi_h|R(v,u)|\psi_h\>=0\quad  \forall h,  \label{eq:HideCondition3}
	\end{align}
	where
	\begin{align}
	R(v,u)&:=\sum_{j,k} R_{kj}(v,u) |\tilde{\phi}_k\>\<\tilde{\phi}_j|,\\
	R_{kj}(v,u)&:=\<b_v^{(k)}|b_u^{(j)}\>-\delta_{uv}\<\psi_k|\psi_j\>. 
	\end{align}

	By assumption the set  $\scrK$ is IC, so \eref{eq:HideCondition3} implies that  $R(v,u)=0$ and $R_{kj}(v,u)=0$, that is, 
	\begin{equation}
	\<b_v^{(k)}|b_u^{(j)}\>=\delta_{uv}\<\psi_k|\psi_j\>,\quad j,k=0,1,\ldots, d-1.
	\end{equation}
	Define $|\bar{b}_u^{(j)}\>:=|u\>\otimes |\psi_j\>$, where $|u\>$ are orthonormal kets in a suitable Hilbert space; then $\<b_v^{(k)}|b_u^{(j)}\>=\<\bar{b}_v^{(k)}|\bar{b}_u^{(j)}\>$. So there exists  an isometry  that maps $|b_u^{(j)}\>$ to $|u\>\otimes |\psi_j\>$. This isometry induces the following transformation, 
	\begin{align}
	|\Psi_h\>\mapsto\sum_{u}  \sqrt{\lambda_u} |a_u\>\otimes (|u\>\otimes |\psi_h\>),	
	\end{align}
	which implies that $M|\psi_h\>\<\psi_h|M^\dag$ have the form in  \eref{eq:antiscrambling}. It follows that $\scrK$ is antiscrambling, which confirms \thref{thm:HideIC}.
\end{proof}

\thmMaskIC*
\begin{proof}[Proof of \thref{thm:MaskIC}]
	Let  $\scrK=\{|\psi_h\>\}_h$ be a set of kets in $\caH$ that is IC; then $|\scrK|\geq d^2$, where $d=\dim(\caH)$. Choose $d$ kets in $\scrK$, say  $|\psi_0\>,|\psi_1\>, \ldots, |\psi_{d-1}\>$, that are linearly independent, so they form a basis  of $\caH$, denoted by $\caB:=\{|\psi_0\>,|\psi_1\>,\ldots, |\psi_{d-1}\>\}$ henceforth. As in the proof of  \thref{thm:HideIC}, let $\tilde{\caB}=\{|\tilde{\phi}_0\>,|\tilde{\phi}_1\>,\ldots, |\tilde{\phi}_{d-1}\>\}$ be the dual basis of $\caB$. 
	
	Suppose on the contrary that $\scrK$ can be masked by the  masker $M$, which maps $\caH$ to a subspace of  $\caH_A\otimes \caH_B$. Let $|\Psi_h\>=M|\psi_h\>$ for all  $h$; then the  two reduced states $\tau_A:=\tr_B (|\Psi_h\>\<\Psi_h|)$ and $\tau_B:=\tr_A (|\Psi_h\>\<\Psi_h|)$ are independent of $h$ by assumption. 
	Suppose $\tau_A$ has spectral decomposition $\tau_A=\sum_u \lambda_u |a_u\>\<a_u|$, where the eigenkets $|a_u\>$ are orthonormal.  Then $|\Psi_h\>$ can be expressed as 
	\begin{equation}
	|\Psi_h\>=\sum_u \sqrt{\lambda_u} |a_u\>|b_u^{(h)}\>,
	\end{equation}
	where the kets $|b_u^{(h)}\>$ for a given $h$ are orthonormal. Accordingly we have
	\begin{equation}
	\tr_A (|\Psi_h\>\<\Psi_h|)=\sum_u \lambda_u |b_u^{(h)}\>\<b_u^{(h)}|=\tau_B\quad \forall h.\label{eq:MaskCondition0}
	\end{equation}
	In addition, each ket $|\psi_h\>$ in $\scrK$ can be expressed as 
	\begin{equation}
	|\psi_h\>=\sum_j |\psi_j\>\<\tilde{\phi}_j|\psi_h\>=\sum_j\mu_{jh} |\psi_j\>,
	\end{equation}
	where  $\mu_{jh}=\<\tilde{\phi}_j|\psi_h\>$, so that
	\begin{align}
	&|\Psi_h\>=\sum_j \mu_{jh} |\Psi_j\>=\sum_{u,j} \mu_{jh} \sqrt{\lambda_u} |a_u\>|b_u^{(j)}\>. 
	\end{align}
	Now the requirement $\tr_A(|\Psi_h\>\<\Psi_h|)=\<\psi_h|\psi_h\>\tau_B$	implies that
	\begin{gather}
	\tr_\caH[(|\psi_h\>\<\psi_h|\otimes \openone_B)\tilde{R}]=0 \quad \forall h , \label{eq:MaskCondition3} 
	\end{gather}
	where $\openone_B$ is the identity operator on $\caH_B$ and 
	\begin{align}
	\tilde{R}&:=\sum_{j,k}|\tilde{\phi}_k\>\<\tilde{\phi}_j|\otimes  \tilde{R}^{(jk)},\\
	\tilde{R}^{(jk)}&:=\sum_u \lambda_u |b_u^{(j)}\>\<b_u^{(k)}|-\<\psi_k|\psi_j\>\tau_B.  \label{eq:Rjk}
	\end{align} 
	Note that $\tilde{R}$ is a hermitian operator acting on  $\caH\otimes \caH_B$.

	By assumption the set $\scrK$ is IC, so \eref{eq:MaskCondition3} implies that  $\tilde{R}=0$, which in turn implies that  $\tilde{R}^{(jk)}=0$ for $j,k=0,1,\ldots, d-1$. In conjunction with \esref{eq:MaskCondition0}{eq:Rjk}, we can deduce the following equality,
	\begin{equation}
	0=\<b_u^{(j)}|\tilde{R}^{(jk)}|b_u^{(k)}\>=\lambda_u -\lambda_u \<\psi_k|\psi_j\>\<b_u^{(j)}|b_u^{(k)}\> \;\; \forall u. \label{eq:Contradiction}
	\end{equation}
	However, this equation cannot hold whenever $\lambda_u>0$ given that 
	$|\<\psi_k|\psi_j\>|<1$ and $|\<b_u^{(j)}|b_u^{(k)}\>|\leq 1$. This contradiction completes the proof of 
	\thref{thm:MaskIC}.
\end{proof}
Here it is instructive to clarify why the above proof does not apply to real quantum mechanics. It is well known that local tomography fails in real quantum mechanics, which means it is in general impossible to determine the state of a composite system by local measurements only \cite{Arak80}. In the case of two qubits for example, each real density matrix can be expressed as 
\begin{align}
\rho=&\frac{1}{4}\bigl[\openone +(a_x\sigma_x +a_z\sigma_z)\otimes \openone_B +\openone_A\otimes (b_x\sigma_x+b_z \sigma_z)\nonumber\\
&+T_{xx} \sigma_x\otimes \sigma_x+T_{xz} \sigma_x\otimes \sigma_z+T_{zx} \sigma_z\otimes \sigma_x\nonumber\\
&+T_{zz} \sigma_z\otimes \sigma_z+T_{yy} \sigma_y\otimes \sigma_y
\bigr],
\end{align}
where $\sigma_x,\sigma_y,\sigma_z$ are the three Pauli matrices,
\begin{equation}
\sigma_x=\begin{pmatrix}
0 &1 \\
1& 0
\end{pmatrix}, \;\;
\sigma_y=\begin{pmatrix}
0 &-\rmi \\
\rmi& 0
\end{pmatrix}, \;\;
\sigma_z=\begin{pmatrix}
1 &0 \\
0&- 1
\end{pmatrix}, 
\end{equation}
and $a_x, a_z, b_x, b_z, T_{xx}, T_{xz}, T_{zx},T_{zz}, T_{yy}
$ are real coefficients. With local measurements only, it is impossible to determine the state $\rho$ completely because such measurements provide no information about  $T_{yy}$. 

Now, let us turn back to the proof of 	\thref{thm:MaskIC}. In real quantum mechanics, due to the failure of local tomography,  \eref{eq:MaskCondition3} alone does not  imply the equality $\tilde{R}=0$, so the reasoning in the proof of \thref{thm:MaskIC} breaks down. Nevertheless, this fact offers little clue on how to construct a masker for the set of real density matrices, so itself does not disprove the no-masking and no-hiding theorems on real quantum mechanics. The maskers we construct in the main text are crucial to settling this issue completely.

\section{\label{sec:HRpropProof}Proofs of \lsref{lem:HRequiDef} and \ref{lem:HRtauTrace}}

\lemHRequiDef*
\begin{proof}[Proof of \lref{lem:HRequiDef}]
	Let
\begin{equation}
V_j:=U_0^\dag U_j,\quad  V(\vec{c}):=\sum_{j=0}^s c_j V_j;
\end{equation}	
 then $V_0=\openone$ and $V(\vec{c})=U_0^\dag U(\vec{c})$. If statement~1 holds, which means $\{V_j\}_{j=1}^{s}$ is a set of $s$ HR matrices, then 
\begin{equation}
V_j^\dag=-V_j,\quad  V_jV_k+V_kV_j=-2\delta_{jk}\openone\;\;  \forall j,k=1,2,\ldots, s.
\end{equation}	
Therefore,
	\begin{align}
	&V(\vec{c})^\dag V(\vec{c})=\sum_{j,k} c_j c_k V_j^\dag V_k \nonumber \\
	&=|\vec{c}|^2 \openone +\sum_{j<k} c_jc_k(V_j^\dag V_k+V_k^\dag V_j)=|\vec{c}|^2 \openone,
	\end{align}
	which implies that $U(\vec{c})^\dag U(\vec{c})=|\vec{c}|^2 \openone$ and confirms the implication $1\imply2$. The implication $2\imply3$ is obvious.

	If statement 3 holds, then the matrix $V(\vec{c})$ is unitary for each normalized real vector $\vec{c}$. Notably, $(\openone+V_j)/\sqrt{2}$ is unitary for $j=1,2,\ldots,s$, which means $V_j^\dag=-V_j$ and $V_j^2=-\openone$. In addition, the requirement that  $(V_j+V_k)/\sqrt{2}$ is unitary for $1\leq j<k\leq s$   implies that 
\begin{equation}
0=V_j^\dag V_k+V_k^\dag V_j=-V_j V_k-V_k V_j. 
\end{equation}	
Therefore, $\{V_j\}_{j=1}^{s}$ is a set of $s$ HR matrices, which confirms the implication $3\imply1$.
\end{proof}

\lemHRtauTrace*

\begin{proof}[Proof of \lref{lem:HRtauTrace}]
Since  $\{U_j\}_{j=1}^s$ is a set of HR matrices, 
we have	
\begin{equation}
U_jU_k+U_kU_j=-2\delta_{jk}\openone\quad \forall 1\leq j<k\leq s. 
\end{equation}	
By the assumptions $s\geq 2$ and $0\leq j<k\leq s$, we can always find a matrix $V$ in  $\scrU:=\{U_j\}_{j=1}^{s}$ that anticommutes with $U_jU_k$. Therefore,
	\begin{align}
	\tr(U_j  U_k \tau^\alpha)&=\tr(V U_j  U_k \tau^\alpha V^\dag)=\tr(V U_j  U_k V^\dag \tau^\alpha )\nonumber\\
	&=-\tr(U_j  U_k \tau^\alpha),
	\end{align}
	which implies \eref{eq:HRtauTrace1}.

	Next, we can prove \eref{eq:HRtauTrace2}   in the case $s\neq 3$. 
	When  $j=j'$ and $k=k'$, \eref{eq:HRtauTrace2} follows from the fact that $(U_j  U_k)^2=-\openone$. When $j\neq j'$ or $k\neq k'$,  we can always find a matrix $V\in \scrU$ that anticommutes with $U_jU_k U_{j'}  U_{k'}$, which implies that $\tr(U_j  U_k U_{j'}  U_{k'} \tau^\alpha)=0$ and confirms \eref{eq:HRtauTrace2}. 
	
	It remains to prove \eref{eq:HRtauTrace3} in the case $s=3$. When $j=j'$ and $k=k'$,	
	\eref{eq:HRtauTrace3} follows from the fact that $(U_j  U_k)^2 =-\openone$ as \eref{eq:HRtauTrace2}. When $j=0$ and $(k, j', k')$ is a permutation of $(1, 2, 3)$, \eref{eq:HRtauTrace3} follows from  the fact that $U_k, U_{j'}, U_{k'}$ anticommute with each other.  When $j'=0$ and $(j,k, k')$ is a permutation of $(1, 2, 3)$, \eref{eq:HRtauTrace3} follows from a similar reasoning. In all other cases, we can always find a matrix $V\in \scrU$ that anticommutes with $U_jU_k U_{j'}  U_{k'}$, which implies that $\tr(U_j  U_k U_{j'}  U_{k'} \tau^\alpha)=0$ and confirms \eref{eq:HRtauTrace3}.
\end{proof}

\section{\label{sec:MaskGenProof}Proof of \thref{thm:MaskGenHR}}

\subsection{Main proof}

\thmMaskGenHR*
Before proving \thref{thm:MaskGenHR}, it is worth pointing out that it	does not cause any loss of generality to assume that 
$\tau_A$ and $\tau_B$ have full rank.  Note that  all states  $|\Psi_j\>$ for $j=0, 1,\ldots, d-1$ have the same reduced density matrix for each party since $M$ is a masker for $\scrD^\rmR(\caH)$. In addition, $\tr_B[M(\rho)]$ for each $\rho\in \scrD(\caH)$ is supported in the support of $\tau_A$. So we can always replace $\caH_A$ with the support of $\tau_A$  if necessary. 	Similarly, we can replace $\caH_B$ with the support of $\tau_B$  if necessary. 
\begin{proof}[Proof of \thref{thm:MaskGenHR}]
	By assumption each $|\Psi_j\>$ for $j=1,2,\ldots, d-1$ has the same reduced states  as $|\Psi_0\>$, so  there exists a unique unitary matrix $U_j$ acting on $\caH_A$ and a unique unitary matrix  $V_j$ acting on $\caH_B$ that satisfy \eref{eq:PsijMaskGen} in \thref{thm:MaskGenHR} according to \lref{lem:SchmidtDecom} below. 
	In addition, $U_j$ and $V_j$ commute with $\tau_A$ and $\tau_B$, respectively. 
	Furthermore, we have  
\begin{equation}
\tr(V_j\tau_B)=\tr(U_j\tau_A)=\<\Psi_0|\Psi_j\>=\<0|j\>=\delta_{j0}=0
\end{equation}	
 since $M$ is an isometry.

	Let $\vec{c}:=(c_0,c_1,\ldots, c_{d-1})$ be an arbitrary normalized real vector; let   $|\psi(\vec{c})\>=\sum_j c_j|j\>$ and $|\Psi(\vec{c})\>=M|\psi(\vec{c})\>$. Then 
	\begin{align}
	|\Psi(\vec{c})\>&=[U(\vec{c})\otimes \openone_B]|\Psi_0\>, \quad U(\vec{c})=\sum_{j=0}^{d-1} c_j U_j,
	\end{align}
	where $U_0=\openone_A$. 
	By assumption $|\Psi(\vec{c})\>$ has the same reduced states  as $|\Psi_0\>$, so  $U(\vec{c})$ is a unitary matrix that  commutes with $\tau_A$ according to \lref{lem:SchmidtDecom} below. Furthermore,   $\{U_j\}_{j=1}^{d-1}$ is a set of $d-1$ HR matrices by \lref{lem:HRequiDef}. Following a similar reasoning,  $\{V_j\}_{j=1}^{d-1}$ is a set of $d-1$ HR matrices.
\end{proof}

\subsection{An auxiliary lemma}

\begin{lemma}\label{lem:SchmidtDecom}
	Suppose
	$|\Psi_1\>, |\Psi_2\>\in \caH_A\otimes \caH_B$ are bipartite pure states 
	that have the same reduced states $\tau_A$ and $\tau_B$ for the two parties, respectively, that is,
	\begin{align}
	\tr_B(|\Psi_1\>\<\Psi_1|)&=\tr_B(|\Psi_2\>\<\Psi_2|)=\tau_A,\\
	\tr_A(|\Psi_1\>\<\Psi_1|)&=\tr_A(|\Psi_2\>\<\Psi_2|)=\tau_B,
	\end{align}	
	where  $\tau_A$ has full rank on $\caH_A$. Then there exists a unique linear operator $U$ acting on $\caH_A$ that satisfies
	\begin{equation}\label{eq:LUtran}
	|\Psi_2\>=(U\otimes \openone_B)|\Psi_1\>;
	\end{equation}
	in addition,  $U$ is  unitary and commutes with $\tau_A$.
\end{lemma}
The last conclusion in \lref{lem:SchmidtDecom}
means $U$ is   block diagonal with respect to the eigenspaces of $\tau_A$.
\begin{proof}
	By assumption both  $|\Psi_1\>$ and $|\Psi_2\>$ are purifications of $\tau_B$, so they can be turned into each other by local unitary transformations on $\caH_A$. In other words, we can find a unitary operator $U$ acting on $\caH_A$ that satisfies \eref{eq:LUtran}.	Let
	\begin{align}
	|\Psi_1\>=\sum_{jk}\Psi_{jk}^{(1)}|jk\>, \quad 
	|\Psi_2\>=\sum_{jk}\Psi_{jk}^{(2)}|jk\>
	\end{align}
	be the decomposition of $|\Psi_1\>$ and $|\Psi_2\>$	in the computational basis and
	define 
	\begin{align}
	M_1=\sum_{jk}\Psi_{jk}^{(1)}|j\>\<k|, \quad 
	M_2=\sum_{jk}\Psi_{jk}^{(2)}|j\>\<k|.
	\end{align}
	Then we have
	\begin{align}
	M_1 M_1^\dag=M_2M_2^\dag=\tau_A, \quad 
	M_2=UM_1, \label{eq:LUtranProof}
	\end{align}
	which imply that
	\begin{equation}
	M_1 M_1^\dag\tau_A^{-1}=\openone_A,\quad 
	U=M_2M_1^\dag\tau_A^{-1},
	\end{equation}
	given that $\tau_A$ has full rank.
	So  there exists only one linear operator $U$ that can satisfy \eref{eq:LUtran}. 
	
	In addition, \eref{eq:LUtranProof} implies that
	\begin{equation}
	U\tau_A U^\dag=\tau_A,
	\end{equation}
	so  the unitary operator $U$ commutes with $\tau_A$.
\end{proof}

\section{\label{sec:MaskGenCorProof}Proofs of  \crsref{cor:MaskSpecCh}-\ref{cor:PurityCon}}

\subsection{Main proofs}
\corMaskSpecCh*
\begin{proof}[Proof of \crref{cor:MaskSpecCh}]
	Suppose $d\geq 3$ and  each $m_o$  is divisible by $\kappa(d)$. Let  $m=\sum_o m_o$ and let $\caH_A$ and $\caH_B$ be two  $m$-dimensional Hilbert spaces, which are isomorphic. Let $|\Psi_0\>\in \caH_A\otimes\caH_B$ be any ket whose reduced state $\tau_A$ for party $A$ has spectrum $\{\lambda_o, m_o\}_o$.  Let $\caH_o$ be the eigenspace of $\tau_A$ associated with $\lambda_o$, then $\caH_o$ has dimension $m_o$.  By \lref{lem:HRnumMax},  we can construct a set $\{V_{o,j}\}_{j=1}^{d-1}$
	of $d-1$ HR matrices
for each subspace $\caH_o$. Let $U_j=\bigoplus_o V_{o,j}$ for $j=1,2,\ldots, d-1$; then  $\{U_j\}_{j=1}^{d-1}$ is  a set of HR matrices that commute with $\tau_A$.

	Define 
	\begin{equation}
	|\Psi_j\>=(U_j \otimes \openone_B )|\Psi_0\>,\quad j=1,2, \ldots, d-1
	\end{equation}
	as in the main text. Then 
	\begin{align}
	\<\Psi_0|\Psi_j\>&=\tr(U_j\tau_A)=0\quad \forall j=1,2,\ldots,d-1,\label{eq:PsijPsi0Inner}\\
	\!\<\Psi_k|\Psi_j\>&=\tr(U_k^\dag U_j\tau_A)=\delta_{jk} \quad \forall j,k=1,2,\ldots, d-1,  \label{eq:PsijPsikInner}
	\end{align}
	given  that $U_k$ anticommutes with $U_j$ and $U_k^\dag U_j$ for $j\neq k$ (cf.~\lref{lem:HRtauTrace}), so the $d$ kets $|\Psi_j\>$ for $j=0,1,\ldots, d-1$ are orthonormal. In addition, $\sum_j c_j U_j$ is  unitary  for any normalized real vector $\vec{c}=(c_j)_j$. So  the isometry $M$ determined by the map	$|j\>\mapsto |\Psi_j\>$ for $j=0, 1,\ldots, d-1$ is a masker for $\scrP^\rmR(\caH)$ and $\scrD^\rmR(\caH)$ with masking  spectrum $\{\lambda_o, m_o\}_o$. 
	
	Conversely, if there exists a masker with masking spectrum $\{\lambda_o,m_o\}_o$, then \thref{thm:MaskGenHR} implies that there exists a set of $d-1$ HR matrices in dimension $m_o$ for each $o$. Therefore, $m_o$ is divisible by $\kappa(d)$ according to \lref{lem:HRnumMax}.

	When $d\geq 2$ and $m_o$ is divisible by $\kappa^\rmR(d)$, a real masker can be constructed as long as $|\Psi_0\>$ is a real ket and $\{V_{o,j}\}_{j=1}^{d-1}$  for each $o$ is a set of real orthogonal HR matrices. In this case, $\tau_A$ is symmetric, while each $U_j$ is antisymmetric, so \esref{eq:PsijPsi0Inner}{eq:PsijPsikInner} hold even when $d=2$.

	 Conversely, the existence of a real masker with masking spectrum $\{\lambda_o,m_o\}_o$ implies that each $m_o$ is divisible by $\kappa^\rmR(d)$ according to a similar reasoning employed above.
	
	It remains to consider complex maskers for a qubit (with $d=2$). Suppose there exists a vector $v=(v_l)_l$ composed of 1 and $-1$  that satisfies the equality $\sum_l v_l\mu_l=0$. Let  $U=\sum_l \rmi v_l |l\>\<l|$ and
	\begin{equation}\label{eq:Maskerd2}
	\begin{aligned}
	|\Psi_0\>&=\sum_l\sqrt{\mu_l}|ll\>,\\ |\Psi_1\>&=(U\otimes \openone_B)|\Psi_0\>=\rmi\sum_l  v_l \sqrt{\mu_l}|ll\>.
	\end{aligned}
	\end{equation}
	Then $|\Psi_0\>, |\Psi_1\>$ are orthonormal,  and the isometry $M$ defined by the map	$|j\>\mapsto |\Psi_j\>$ for $j=0, 1$ is a masker for $\scrD^\rmR(\caH)$. 
	
	Conversely, suppose $M$ is a masker for the set $\scrD^\rmR(\caH)$ with masking spectrum $\{\mu_l\}_l$. Let $|\Psi_0\>=M|0\>$ and $\tau_A=\tr_B(|\Psi_0\>\<\Psi_0|)$. Without loss of generality, we can assume that $\tau_A$ has full rank. According to \thref{thm:MaskGenHR}, there exists an HR matrix $U$ that commutes with $\tau_A$ and satisfies $\tr(U\tau_A)=0$. By a suitable choice of orthonormal basis if necessary, we can assume that both $U$ and $\tau_A$ are diagonal. The diagonal entries of $\tau_A$ happen to be $\mu_l$, while the diagonal entries of $U$ have the form $ v_l \rmi$ with $v_l=\pm1$, 
	given that  each eigenvalue of $U$ is either $\rmi$ or $-\rmi$. Therefore,  $\rmi\sum_l v_l\mu_l=\tr(U\tau_A)=0$, which completes the proof of \crref{cor:MaskSpecCh}.
\end{proof}

\corMaskDim*
\begin{proof}[Proof of \crref{cor:MaskDim}]
	\Crref{cor:MaskDim} is an immediate consequence of \crref{cor:MaskSpecCh}.
\end{proof}

\begin{widetext}

\corPurityCon*

	\begin{proof}[Proof of \crref{cor:PurityCon}]
		Without loss of generality, we can assume that $\tau_A$ and $\tau_B$ have full rank. Let $|\Psi_j\>=M|j\>$, then	$|\Psi_j\>$ have the form in \eref{eq:PsijMaskGen} according to \thref{thm:MaskGenHR}, so $\varrho_A$ and $\varrho_B$ have the form in \esref{eq:varrhoA}{eq:varrhoB} according to \lref{lem:varrhoAB} below.
	In addition, $\varrho_A$ and $\varrho_B$ commute with $\tau_A$ and $\tau_B$, respectively. 
Let $P$ be the  projector onto the eigenspace of $\tau_A$ associated with any nonzero eigenvalue $\lambda$; then $P$ commutes with $\varrho_A$. Meanwhile, $\tr(P U_j U_k)=0$ for $0\leq j<k<d$ since we can always find a matrix in $\{U_j\}_{j=1}^{d-1}$ that anticommutes with $U_j U_k$ (note that this conclusion does not hold when $d=2$, and that is why we assume $d\geq 3$ in \crref{cor:PurityCon}). Therefore,
		\begin{align}
	\!\!	\tr(P\varrho_A)&=\tr(P\tau_A)+\sum_{0\leq j <k<d} (\rho_{kj}-\rho_{jk})\tr(PU_j  U_k \tau_A)=\tr(P\tau_A)+\sum_{0\leq j <k<d} (\rho_{kj}-\rho_{jk})\lambda\tr(PU_j  U_k)=\tr(P\tau_A),
		\end{align}
		which implies that $P\tau_A$ is majorized by $P\varrho_A$ given that $P\tau_A=\lambda P$ is proportional to a projector. As a corollary,  $\tau_A$ is majorized by $\varrho_A$. By the same token $\tau_B$ is majorized by $\varrho_B$.

		When $d\neq 4$, by \lref{lem:HRtauTrace} in the main text and \lref{lem:varrhoAB} below,  the purity of $\varrho_A$ can be computed as follows,
		\begin{align}
		\tr(\varrho_A^2)&=\tr(\tau_A^2)+\sum_{\substack{0\leq j<k<d\\0\leq j'<k'<d}}
		(\rho_{kj}-\rho_{jk})(\rho_{k'j'}-\rho_{j'k'})\tr(U_j  U_k U_{j'}  U_{k'} \tau_A^2)\nonumber\\
		&=\wp\Biggl[1  - \sum_{0\leq j <k<d}(\rho_{kj}-\rho_{jk})^2\Biggr]=\wp\bigr[1+\tr(\rho^2)-\tr\bigl(\rho\rho^\rmT\bigr)\bigr]=\frac{1}{2}\wp\bigr(2+\|\rho-\rho^\rmT\|_{\hs}^2\bigr), \label{eq:rhoApurity}
		\end{align}
		where $\wp=\tr(\tau_A^2)=\tr(\tau_B^2)$ is the masking purity. Here the last equality follows from the fact that 
		\begin{equation}
		\|\rho-\rho^\rmT\|_{\hs}^2=2\bigl[\tr(\rho^2)-\tr\bigl(\rho\rho^\rmT\bigr)\bigr]. 
		\end{equation}
		By the same token we can deduce
		\begin{equation}
		\tr(\varrho_B^2)=\wp\bigr[1+\tr(\rho^2)-\tr\bigl(\rho\rho^\rmT\bigr)\bigr]=\frac{1}{2}\wp\bigr(2+\|\rho-\rho^\rmT\|_{\hs}^2\bigr),
		\end{equation}
		which, together with \eref{eq:rhoApurity}, confirms \esref{eq:purityABsum}{eq:purityAB}.
		
		\Eref{eq:purityABsumLBUB} is a corollary of  \eref{eq:purityABsum} and the following equation, 
		\begin{align}
		\frac{4}{d}\scrI_\rmR(\rho)^2\leq 	\|\rho-\rho^\rmT\|_{\hs}^2\leq 2\scrI_\rmR(\rho)^2. \label{eq:rhorhoTHS}
		\end{align}
		\Eref{eq:rhorhoTHS} in turn follows from 	the equality $\scrI_\rmR(\rho)=\|\rho-\rho^\rmT\|_1/2$ in \eref{eq:RoIvalue} \cite{WuKRS21} and the inequalities
		\begin{equation}
		\frac{1}{d}\|\rho-\rho^\rmT\|_1^2\leq \|\rho-\rho^\rmT\|_{\hs}^2\leq \frac{1}{2}\|\rho-\rho^\rmT\|_1^2.
		\end{equation}	
		Here the first inequality is well known given that $\rho$ is a density matrix on a Hilbert space of dimension $d$; the second inequality follows from the fact that $\rho-\rho^\rmT$ is hermitian and antisymmetric, which means its nonzero eigenvalues form pairs of opposite signs. Similarly, 	\eref{eq:purityABLBUB} is a corollary of  \esref{eq:purityAB}{eq:rhorhoTHS}.

		When $d=4$, by \lsref{lem:HRtauTrace} and \ref{lem:varrhoAB}, the purity of $\varrho_A$ reads
		\begin{align}
		&\tr(\varrho_A^2)=\tr(\tau_A^2)+\sum_{\substack{0\leq j<k<d\\0\leq j'<k'<d}}
		(\rho_{kj}-\rho_{jk})(\rho_{k'j'}-\rho_{j'k'})\tr(U_j  U_k U_{j'}  U_{k'} \tau_A^2)\nonumber\\
		&=\wp\Biggl[1  - \sum_{0\leq j <k\leq 3}(\rho_{kj}-\rho_{jk})^2\Biggr]+\sum_{\substack{0\leq j<k\leq 3\\0\leq j'<k'\leq 3}}
		(\delta_{j=0}\epsilon_{kj'k'}
		+\delta_{j'=0}\epsilon_{j k k'})(\rho_{kj}-\rho_{jk})(\rho_{k'j'}-\rho_{j'k'})\tr(U_1U_2U_3\tau_A^2)\nonumber\\
		&=\wp\bigr[1+\tr(\rho^2)-\tr\bigl(\rho\rho^\rmT\bigr)\bigr]+\sum_{\substack{0\leq j<k\leq 3\\0\leq j'<k'\leq 3}}
		(\delta_{j=0}\epsilon_{kj'k'}
		+\delta_{j'=0}\epsilon_{j k k'})(\rho_{kj}-\rho_{jk})(\rho_{k'j'}-\rho_{j'k'})\tr(U_1U_2U_3\tau_A^2). \label{eq:varrhoApurityd4}
		\end{align}
		Similarly,  the purity of $\varrho_B$ reads
		\begin{align}
		&\tr(\varrho_B^2)=\wp\bigr[1+\tr(\rho^2)-\tr\bigl(\rho\rho^\rmT\bigr)\bigr]+\sum_{\substack{0\leq j<k\leq 3\\0\leq j'<k'\leq 3}}
		(\delta_{j=0}\epsilon_{kj'k'}
		+\delta_{j'=0}\epsilon_{j k k'})(\rho_{kj}-\rho_{jk})(\rho_{k'j'}-\rho_{j'k'})\tr(V_1V_2V_3\tau_B^2)\nonumber\\
		&=\wp\bigr[1+\tr(\rho^2)-\tr\bigl(\rho\rho^\rmT\bigr)\bigr]-\sum_{\substack{0\leq j<k\leq 3\\0\leq j'<k'\leq 3}}
		(\delta_{j=0}\epsilon_{kj'k'}
		+\delta_{j'=0}\epsilon_{j k k'})(\rho_{kj}-\rho_{jk})(\rho_{k'j'}-\rho_{j'k'})\tr(U_1U_2U_3\tau_A^2), \label{eq:varrhoBpurityd4}
		\end{align}
		which, together with \eref{eq:varrhoApurityd4}, implies  \eref{eq:purityABsum}.  The second equality in \eref{eq:varrhoBpurityd4} is a consequence of the equality  $\tr(V_1V_2V_3\tau_B^2)=-\tr(U_1U_2U_3\tau_A^2)$, which can be derived as follows. From  \eref{eq:PsijMaskGen} we can deduce that 
		\begin{equation}
		(\openone_A\otimes V_1 V_2 V_3)|\Psi_0\> =(U_3 U_2U_1\otimes \openone_B)|\Psi_0\>,
		\end{equation}
		which implies that
		\begin{align}
		\tr(V_1V_2V_3\tau_B^2)=\tr(U_3U_2U_1\tau_A^2)=-\tr(U_1U_2U_3\tau_A^2).
		\end{align}
		Here the first equality follows from \lref{lem:UVAB} below, and the second one follows from the fact that $U_1, U_2, U_3$ anticommute with each other. 
		
		\Eref{eq:purityABsumLBUB} is still a corollary of  \esref{eq:purityABsum}{eq:rhorhoTHS}. This observation completes the proof of \crref{cor:PurityCon}.
	\end{proof}
\end{widetext}

\bigskip
\subsection{Two auxiliary lemmas}
\begin{lemma}\label{lem:varrhoAB}
	Suppose $M$ is the masker for $\scrD^\rmR(\caH)$ in \thref{thm:MaskGenHR}. Let $\rho\in \scrD(\caH)$,  $\varrho=M(\rho)$, $\varrho_A=\tr_B(\varrho)$, and $\varrho_B=\tr_A(\varrho)$. Then 
	\begin{align}
	\varrho_A&=\tau_A+\sum_{0\leq j <k<d} (\rho_{kj}-\rho_{jk})U_j  U_k \tau_A, \label{eq:varrhoA}\\
	\varrho_B&=\tau_B+\sum_{0\leq j <k<d} (\rho_{kj}-\rho_{jk})V_j  V_k \tau_B, \label{eq:varrhoB}
	\end{align}
	where $U_0=\openone_A$, $V_0=\openone_B$, and  $U_j, V_j$ for $j=1,2,\ldots, d-1$ are determined by \eref{eq:PsijMaskGen}.
\end{lemma}

\begin{proof}
	By assumption we have
	\begin{align}
	\varrho&=M(\rho)=\sum_{j,k} \rho_{jk} M|j\>\<k|M^\dag=\sum_{j,k} \rho_{jk}|\Psi_j\>\<\Psi_k|\nonumber\\
	&=\sum_{j,k} \rho_{jk}(U_j\otimes \openone_B)|\Psi_0\>\<\Psi_0|(U_k^\dag\otimes \openone_B). 
	\end{align}
	Therefore,
	\begin{align}
	&\varrho_A=\tr_B(\varrho)=\sum_{j,k}\rho_{jk} U_j\tau_A  U_k^\dag=\sum_{j,k}\rho_{jk} U_j  U_k^\dag \tau_A\nonumber\\
	&=\tau_A+\sum_{j\neq k}\rho_{jk} U_j  U_k^\dag \tau_A\nonumber\\
	&=\tau_A+\sum_{0\leq j <k<d} (\rho_{kj}-\rho_{jk})U_j  U_k \tau_A. \label{eq:MaskMix}
	\end{align}
	Here the third equality follows from the fact that $U_k$ and $U_k^\dag$ commute with $\tau_A$. The last equality follows from the fact that $U_0=\openone_A$ and $\{U_j\}_{j=1}^{d-1}$ is a set  of HR matrices on $\caH_A$, which means 
	\begin{equation}
	U_j U_k +U_k U_j=-2\delta_{jk}\openone_A\quad \forall j,k=1,2,\ldots, d-1.
	\end{equation}
	\Eref{eq:varrhoB} can be proved in a similar way.
\end{proof}

\begin{lemma}\label{lem:UVAB}
	Let $|\Psi_0\>$ be a bipartite ket in the Hilbert space  $\caH_A\otimes \caH_B$ with reduced states  $\tau_A=\tr_B(|\Psi_0\>\<\Psi_0|)$ and $\tau_B=\tr_A(|\Psi_0\>\<\Psi_0|)$. Suppose $U$ and $V$ are unitary operators on $\caH_A$ and $\caH_B$ that satisfy
	\begin{equation}\label{eq:UVAB}
	(U\otimes \openone_B)|\Psi_0\>=(\openone_A\otimes V)|\Psi_0\>;
	\end{equation}
	then 
	\begin{equation}\label{eq:UVABtauTrace}
	\tr(\tau_A^\alpha U)=\tr(\tau_B^\alpha V)\quad \forall \alpha>0.
	\end{equation} 
	Suppose $P_A$ and $P_B$ are the projectors onto the eigenspaces of $\tau_A$  and $\tau_B$, respectively, which are associated with a same nonzero eigenvalue $\lambda$;  then 
	\begin{align}
	\tr(P_A U)=\tr(P_B V). \label{eq:UVABtraceP}
	\end{align}
\end{lemma}

\begin{proof}
	We first prove \eref{eq:UVABtraceP}.	
	\Eref{eq:UVAB} implies that $U$ commutes with $\tau_A$ and $P_A$, while $V$ commutes with $\tau_B$ and $P_B$. In addition, by assumption we have 
	\begin{align}
	(P_A\otimes \openone_B)|\Psi_0\>&=(\openone_A\otimes P_B)|\Psi_0\>,\\
	(P_AU\otimes \openone_B)|\Psi_0\>&=(P_A\otimes \openone_B)(\openone_A\otimes V)|\Psi_0\>\nonumber\\
	&=(\openone_A\otimes P_B V)|\Psi_0\>. 
	\end{align}
	Therefore,
	\begin{align}
	&	\lambda\tr(P_A  U)=\tr(P_A  U\tau_A)=\<\Psi_0|(P_A U\otimes \openone_B)|\Psi_0\>\nonumber\\
	&
	=\<\Psi_0|(\openone_A\otimes P_B V)|\Psi_0\>=\tr(P_B V\tau_B)=\lambda\tr(P_B V),
	\end{align}
	which implies \eref{eq:UVABtraceP}.
	
	\Eref{eq:UVABtauTrace} is a corollary of \eref{eq:UVABtraceP} given that $\tau_A$ and $\tau_B$ have the same nonzero spectrum. When $\alpha=1$, \eref{eq:UVABtauTrace}  can also be derived directly as follows,
	\begin{align}
	\tr(\tau_A U)&=\<\Psi_0|(U\otimes \openone_B)|\Psi_0\> =\<\Psi_0|(\openone_A\otimes V)|\Psi_0\>\nonumber\\
	&=\tr(\tau_B V).
	\end{align}
\end{proof}

\section{\label{sec:MaskEntProof}Proof of \crref{cor:MaskEnt}}

\corMaskEnt*
\begin{proof}[Proof of \crref{cor:MaskEnt}]
	Note that any masker for $\scrD^\rmR(\caH)$	is a masker for $\scrP^\rmR(\caH)$, and vice versa, so we have $E(\scrD^\rmR(\caH))=E(\scrP^\rmR(\caH))$ according to \lref{lem:MaskEnt}.
	Suppose $M: \caH\mapsto \caH_A\otimes\caH_B$ is a masker for  $\scrD^\rmR(\caH)$. Let $|\Psi_0\>=M|0\>$; then \lref{lem:MaskEnt} implies that
	\begin{equation}\label{eq:MaskEntRproof}
	E(\scrD^\rmR(\caH))=E(|\Psi_0\>\<\Psi_0|)\geq E(|\Phi(\tilde{\kappa}(d))\>),
	\end{equation}
	where $\tilde{\kappa}(d):= \max\{\kappa(d),2\}$ (cf.~ \crref{cor:MaskDim}). 
	Here the inequality follows from the fact that the nonzero spectrum of $\tr_B(|\Psi_0\>\<\Psi_0|)$ (that is, the masking spectrum of $M$) is majorized by the nonzero spectrum of $\tr_B(|\Phi(\tilde{\kappa}(d))\>\<\Phi(\tilde{\kappa}(d))|)$ according to  \crref{cor:MaskSpecCh}, so $|\Psi_0\>$ can be turned into $|\Phi(\tilde{\kappa}(d))\>$ by LOCC according to the  majorization criterion of Nielsen~\cite{Niel99}. In addition, the inequality in \eref{eq:MaskEntRproof} can be saturated by choosing a suitable masker  according to  \crref{cor:MaskSpecCh}. This observation confirms \eref{eq:MaskEntR}. \Eref{eq:MaskEntRR} can be proved using a similar reasoning.
\end{proof}

\section{\label{sec:MaskMaxProof}Proofs of \thsref{thm:MaskMax} and \ref{thm:MaskEnt}}

\thmMaskMax*

\begin{proof}[Proof of \thref{thm:MaskMax}] Let  $\rho\in \scrD(\caH)$ be an arbitrary density matrix; suppose the set  $\scrD^\rmR(\caH)\cup \{\rho\}$ can be  masked by the isometry  $M: \caH\mapsto\caH_A\times \caH_B$. Then $M$ is also a maker for $\scrD^\rmR(\caH)$.  Let  $\tau_A$ and $\tau_B$ be the  common reduced density matrices. Without loss of generality, we can assume that $\tau_A$ and $\tau_B$ have full rank.  Let 	 $\varrho=M(\rho)$, $\varrho_A=\tr_B(\varrho)$, and $\varrho_B=\tr_A(\varrho)$.
Then we have $\varrho_A=\tau_A$ and $\varrho_B=\tau_B$, so that  
	\begin{equation}
	\tr(\varrho_A^2)+\tr(\varrho_B^2)=\tr(\tau_A^2)+\tr(\tau_B^2)=2\tr(\tau_A^2).
	\end{equation}	
	When  $d\geq 3$, in conjunction with \crref{cor:PurityCon}, this equation implies that $\tr(\rho\rho^\rmT)=\tr(\rho^2)$, so that $\rho\in \scrD^\rmR(\caH)$. 	Therefore,  $\scrD^\rmR(\caH)$ is a maximal maskable set.

	When  $d=2$, the density matrix $\varrho_A$ has the form 
	\begin{equation}
	\varrho_A=\tau_A+(\rho_{10}-\rho_{01})U\tau_A
	\end{equation}
	according to \lref{lem:varrhoAB}, where $U$ is a unitary matrix that satisfies $U^2=-\openone_A$. So the equality $\varrho_A=\tau_A$ implies that $\rho_{10}=\rho_{01}$ and that  $\rho\in \scrD^\rmR(\caH)$. 	Therefore, $\scrD^\rmR(\caH)$ is  a maximal maskable set. Alternatively, this conclusion follows from \crref{cor:MaskQubit} (cf. \rscite{LianLF19,DingH20}). This observation 	
	completes the proof of \thref{thm:MaskMax}.

For pure states, \thref{thm:MaskMax} can be proved in a simpler way, which is less technical and more instructive. 
	Let $|\psi\>$ be any ket in $\caH$ that is not proportional to a real ket. Let $\scrS$ be the union of $\scrP^\rmR(\caH)$ and $\{|\psi\>\<\psi|\}$. To prove that $\scrP^\rmR(\caH)$	is a maximal maskable set within $\scrP(\caH)$, it suffices to prove that $\scrS$ is not maskable.

	By assumption $|\psi\>$ has the form 
	$|\psi\>=c_1|\psi_1\>+\rmi c_2|\psi_2\>$, where $|\psi_1\>, |\psi_2\>$ are two real kets that are linearly independent, and $c_1, c_2$ are nonzero real coefficients that satisfy the normalization condition $c_1^2+c_2^2=1$. Let $\caH_2$ be the span of $\{|\psi_1\>,|\psi_2\>\}$ and let
	$\scrS_2$ be the set of all real pure states on $\caH_2$. Then we have $|\psi\>\in \caH_2$ and  $\scrS_2\cup \{|\psi\>\<\psi|\}\subseteq \scrS$. Geometrically,  $\scrS_2$ is a great circle on the Bloch sphere associated with $\caH_2$, while $|\psi\>\<\psi|$ is a point  on the sphere, but not on this great circle. Therefore,   $\scrS_2 \cup \{|\psi\>\<\psi|\}$ is not maskable according to \crref{cor:MaskQubit}, which implies that $\scrS$ is not maskable.
\end{proof}

\thmMaskEnt*

\begin{proof}[Proof of \thref{thm:MaskEnt}]
	Let $\tau_A=\tr_B(|\Psi_0\>\<\Psi_0|)$. Let $\rho$ be an arbitrary state in   $\scrD(\caH)$, $\varrho=M(\rho)$, and  $\varrho_A=\tr_B(\varrho)$.	According to \crref{cor:PurityCon}, $\tau_A$ is majorized by $\varrho_A$. If $\rho$ is pure, so that $\varrho$ is also pure, then $\varrho$ can be created from $|\Psi_0\>\<\Psi_0|$ by LOCC according to the  majorization criterion \cite{Niel99}. The same conclusion applies when $\rho$ is mixed, in which case $\rho$ can be expressed as a convex mixture of pure states and so can $\varrho$ accordingly.
	
By  the above discussion, we have $E(\varrho)\leq E(|\Psi_0\>\<\Psi_0|)$ for any entanglement monotone $E$ since $|\Psi_0\>\<\Psi_0|$ can be turned into $\varrho$ by LOCC. Next, suppose $E$ is a convex strict monotone. If $\rho,\varrho$ are pure and the upper bound $E(\varrho)\leq  E(|\Psi_0\>\<\Psi_0|)$ is saturated,
	then $\varrho$ and $|\Psi_0\>\<\Psi_0|$ can be turned into each other by LOCC. Therefore,  $\varrho_A$,  $\varrho_B$, $\tau_A$, and $\tau_B$ have the same nonzero spectrum  according to the  majorization criterion \cite{Niel99} and thus have the same purity. Now \crref{cor:PurityCon}  implies that $\rho$ is a real density matrix, that is, $\rho\in \scrP^\rmR(\caH)$.

	If $\rho$ is mixed, then it can be expressed as a convex mixture of pure states  $\rho=\sum_r a_r|\psi_r\>\<\psi_r|$ with $a_r>0$ and $|\psi_r\>\<\psi_r|\notin \scrP^\rmR(\caH)$  for at least one $r$. Accordingly, we have $\varrho=\sum_r a_rM |\psi_r\>\<\psi_r|M^\dag$,
\begin{align}
E(M|\psi_r\>\<\psi_r| M^\dag)\leq E(|\Psi_0\>\<\Psi_0|)\quad \forall r,
\end{align}	
 and the inequality is strict for at least one $r$. Therefore,
	\begin{equation}
	E(\varrho)\leq \sum_r a_r E(M|\psi_r\>\<\psi_r| M^\dag)<E(|\Psi_0\>\<\Psi_0|),
	\end{equation}
	so the inequality $E(\varrho)\leq  E(|\Psi_0\>\<\Psi_0|)$ cannot be saturated when $\rho$ is mixed, which  completes the proof.
\end{proof}

\section{\label{sec:PhaseRealProof}Proofs of \psref{pro:PhaseRoReal} and \ref{pro:RealToPhase}}

\proPhaseRoReal*
\begin{proof}[Proof of \pref{pro:PhaseRoReal}]
	Note that $\scrP(\vec{c})$ and $\scrP(\vec{c}')$ are equivalent whenever $\vec{c}'$ is a permutation of  $\vec{c}$. So we can assume that $0<c_0\leq c_1\leq \cdots\leq c_{d-1}<1$ without loss of generality, which implies that $0<c_0^2 \leq 1/d$. Let 
	\begin{align}
	|\psi_k\rangle =c_0\rme^{2k\pi\rmi/3}|0\rangle+\sum_{j=1}^{d-1} c_j  |j\rangle, \quad k=0,1,2;
	\end{align}
	then $|\psi_k\rangle\in \scrK(\vec{c})$ and  $|\psi_k\rangle\langle \psi_k|\in \scrP(\vec{c})$. Suppose $\scrP(\vec{c})$ is equivalent to a subset of $\scrP^\rmR(\caH')$; then the following triple product must be real,
	\begin{align}
	\tr[(|\psi_0\rangle\langle\psi_0|) (|\psi_1\rangle\langle\psi_1|) (|\psi_2\rangle\langle\psi_2|)]=(c_0^2\rme^{2\pi\rmi/3}+1-c_0^2)^3,
	\end{align}
	since the triple product is invariant under isometry. 
	In conjunction with the condition $0<c_0^2 \leq 1/d\leq 1/2$, we can deduce that $d=2$ and $c_0^2=1-c_0^2$, which implies that $c_0=c_1=1/\sqrt{2}$. This observation completes the proof of \pref{pro:PhaseRoReal}.
\end{proof}

\proRealToPhase*

\begin{proof}[Proof of \pref{pro:RealToPhase}]
	Suppose on the contrary that the set $\scrP^\rmR(\caH')$ with $\dim(\caH')\geq 3$ is equivalent to a subset of $\scrP(\vec{c})$ under the isometry $S$. Let $|\psi_k\rangle=S|k\rangle$ for $k=0,1,2$; then $|\psi_k\rangle$ have the form
	\begin{align}
	|\psi_0\rangle &=\sum_j c_j \rme^{\rmi\alpha_j} |j\rangle,\\
	|\psi_1\rangle &=\sum_j c_j \rme^{\rmi\beta_j} |j\rangle,\\
	|\psi_2\rangle &=\sum_j c_j \rme^{\rmi\gamma_j} |j\rangle,
	\end{align}
	where $0\leq \alpha_j,\beta_j,\gamma_j<2\pi$ for $j=0,1,\ldots,d-1$. 
	Moreover, we can assume that $\alpha_j=0$ for all $j$ without loss of generality. Let 
	\begin{equation}
	|\psi_{kl}\rangle=\frac{1}{\sqrt{2}}S(|k\rangle+|l\rangle),\quad 0\leq k<l\leq 2;
	\end{equation}
	then $|\psi_{kl}\rangle\in \scrK(\vec{c})$ by assumption, which implies that
	\begin{gather}
	\biggl|\frac{1+\rme^{\rmi\beta_j}}{\sqrt{2}}\biggr|^2=\biggl|\frac{1+\rme^{\rmi\gamma_j}}{\sqrt{2}}\biggr|^2=\biggl|\frac{1+\rme^{\rmi(\beta_j-\gamma_j)}}{\sqrt{2}}\biggr|^2=1,\\
	\cos(\beta_j)=\cos(\gamma_j)=\cos(\beta_j-\gamma_j)=0\quad \forall j. 
	\end{gather}
	However, the last equation can never hold. This contradiction shows that $\scrP^\rmR(\caH')$ cannot be  equivalent to a subset of $\scrP(\vec{c})$ except when  $\dim(\caH')=2$, which confirms \pref{pro:RealToPhase}.
\end{proof}

\bibliography{all_references}

\end{document}